\documentclass[a4paper,twocolumn,11pt]{quantumarticle}
\pdfoutput=1

\usepackage[utf8]{inputenc}
\usepackage[english]{babel}
\usepackage[T1]{fontenc}

\usepackage{amsmath}
\usepackage{amsthm}
\usepackage{amsfonts}
\usepackage{amssymb}
\usepackage{braket}

\usepackage{makecell}
\usepackage{multirow}

\usepackage{xcolor}

\usepackage{hyperref}       
\hypersetup{colorlinks=true}
\usepackage{url}

\usepackage{graphicx}
\usepackage{caption}
\usepackage{float}
\usepackage{subcaption}

\usepackage{tikz}
\usetikzlibrary{shadows}
\usetikzlibrary{matrix, positioning}

\usepackage[numbers,compress]{natbib}
\usepackage[compat=0.6]{yquant}
\useyquantlanguage{groups}
\usetikzlibrary{fit,quotes,positioning}
\pgfdeclareshape{yquant-multictrlshape}{
\inheritsavedanchors[from=yquant-circle]
\foreach \anc in {center, north, north east, east, south east, south, south west, west, north west} {
	\inheritanchor[from=yquant-circle]{\anc}
}
\inheritanchorborder[from=yquant-circle]
\backgroundpath{
	\pgfpathmoveto{\pgfqpoint{-0.707107\dimexpr\xradius\relax}{-0.707107\dimexpr\yradius\relax}}
	\pgfpathlineto{\pgfqpoint{0.707107\dimexpr\xradius\relax}{0.707107\dimexpr\yradius\relax}}
	\pgfpathellipse{\pgfpointorigin}
	{\pgfqpoint{\xradius}{0pt}}
	{\pgfqpoint{0pt}{\yradius}}
}
\inheritclippath[from=yquant-circle]
}
\yquantset{multictrl/.style={every control/.style={shape=yquant-multictrlshape,radius=1.05mm}}}

\usepackage{orcidlink}

\newtheorem{theorem}{Theorem}
\newtheorem{definition}{Definition}

\newtheorem{lemma}{Lemma}
\newtheorem{corollary}{Corollary}
\newtheorem{proposition}{Proposition}

\begin{document}

\title{Resource-Tunable Quantum Circuit Implementation of Nonlinear Element-Wise Transformations}

\author{Chunlin Yang}
\affiliation{School of Mathematical Sciences, Harbin Engineering University, Harbin 150001, China}
\orcid{0009-0005-9716-5411}

\author{Yuqi Li}
\affiliation{School of Mathematical Sciences, Harbin Engineering University, Harbin 150001, China}

\author{Hongmei Yao}
\email{hongmeiyao@hrbeu.edu.cn}
\orcid{0000-0001-7772-8006}
\affiliation{School of Mathematical Sciences, Harbin Engineering University, Harbin 150001, China}

\author{Zhaobing Fan}
\affiliation{School of Mathematical Sciences, Harbin Engineering University, Harbin 150001, China}

\begin{abstract}
    Nonlinear transformations are important in quantum machine learning and other quantum numerical applications, which can be reduced to the implementation of element-wise polynomials through polynomial approximation.
    Existing works establish the feasibility of transformations and improve query or auxiliary-space efficiency, but circuit depth remains linear in scale with degree.
    In this work, we develop a resource-tunable quantum framework that decomposes a $(2^d-1)$-degree polynomial into $m$ lower-degree factors, with $m$ controlling the degree of parallelism.
    By varying $m$ from $1$ to $2^d-1$, the framework provides a trade-off among query complexity, additional gate complexity, ancilla count, and normalization, ranging from a binary-tree implementation to a $\mathcal{O}(n+d)$-depth construction with highly parallel oracle queries that match the query lower bound. 
    The effects of polynomial approximation and implementation errors are further analyzed. We demonstrate the framework for nonlinear activation functions such as sigmoid and tanh, as well as pixel-wise image transformations.
\end{abstract}

\maketitle

\section{Introduction}\label{sec: introduction}
Nonlinear transformations are central to many computational models, most notably neural networks, where nonlinear activation functions are essential for expressiveness~\cite{cybenko1989approximation, hornik1989multilayer, pinkus1999approximation}.
The prospect of exploiting quantum computation for machine learning and other numerical tasks has consequently motivated a broad range of quantum neural-network and quantum linear algebra approaches~\cite{schuld2014quest, benedetti2019parameterized, cerezo2021variational}.
A fundamental obstacle, however, arises when a nonlinear function must act directly on quantum-encoded data: admissible quantum dynamics is linear on state space and cannot realize an arbitrary entry-wise nonlinear map on amplitudes for arbitrary inputs~\cite{abrams1998nonlinear}.
Such a nonlinearity must instead be embedded in a larger linear process, typically via ancillary registers and post-selection. 
The full process remains linear, and nonlinearity appears only on the post-selected branch, where the conditional state depends nonlinearly on the input at the cost of an input-dependent failure probability and validity restricted to a limited class of inputs. 
This tension between linear evolution and nonlinear data processing is not confined to quantum neural networks; it arises whenever nonlinear functions are applied entry-wise to quantum-encoded data.

A natural way to formulate this problem is through polynomial approximation~\cite{pinkus2000weierstrass, mergelyan1951, mergelyan1952} of the scalar function to be applied. 
Element-wise transformations of this kind can be formulated for quantum-accessible data in either of the two standard representations~\cite{guo2025quantum}: the amplitudes of a quantum state and the matrix by block encoding~\cite{gilyen2019quantum}.
Let $A=(a_{ij})\in\mathbb{C}^{2^n\times 2^n}$ be a matrix encoded by a block encoding with bounded entries, and consider the entry-wise transformation
\begin{equation*}
    A \longmapsto f(A)=\left(f(a_{ij})\right),
\end{equation*}
where $f:\mathbb{C}\rightarrow\mathbb{C}$ is a nonlinear function that can be approximated by a polynomial. Approximating $f$ by the polynomial $P(x)=\sum_{k}c_kx^k$ reduces the nonlinear transformation to the implementation of
\begin{equation*}
    P(A) = \left(\sum_{k}c_ka_{ij}^k\right),
\end{equation*}
This object should be distinguished from the ordinary matrix polynomial $\sum_k c_kA^k$: the latter is governed by matrix multiplication, whereas $P(A)$ acts independently on each matrix entry and is called a Hadamard matrix function~\cite{horn1994topics}. 
Once the target nonlinear function has been reduced to this form, the central algorithmic problem is to construct a block encoding of a Hadamard matrix function efficiently.

Several approaches have been developed for nonlinear transformations of states, including quantum analog-to-digital conversion, which converts transformations on amplitudes into transformations on bit strings~\cite{mitarai2019quantum}; the weighted-state framework with its quantum Hadamard product subroutine for preparing powers and polynomials of state amplitudes~\cite{holmes2023nonlinear}; and constructions based on quantum singular value transformation (QSVT)~\cite{gilyen2019quantum,martyn2021grand} for state amplitudes, which transform states into diagonal matrices~\cite{rattew2023nonlinear, guo2024nonlinear}. 
More recently, element-wise transformations of general block-encoded matrices were constructed by combining Hadamard products with linear combinations of unitaries (LCU)~\cite{guo2025quantum}. A subsequent work~\cite{rossi2026quantum} improved the construction of~\cite{guo2025quantum}, reducing the auxiliary space required for such transforms to logarithmic in the polynomial degree while retaining linear depth in the degree.
Across this line of work, the relevant resources are now comparatively well understood—query complexity, auxiliary space, subnormalization, and success probability—but the circuit depth of every known construction is linear in the degree.
No systematic method has been described for reducing this depth, and the trade-off between depth and other resources has not been characterized. This motivates the question addressed in this work: whether the depth of implementation of a Hadamard matrix function can be reduced below linear in degree, and at what cost in auxiliary space, success probability, and normalization.

In this work, we develop a resource-tunable quantum framework for Hadamard matrix functions with explicit circuit construction. 
The central idea is to decompose a $(2^d-1)$-degree polynomial into a Hadamard product of $m$ lower-degree factors. We design a recursive binary-tree construction to implement these lower-degree factors, which reuses intermediate Hadamard powers and reduces the required auxiliary qubits at the cost of increased depth. The value of $m$ provides tunable trade-offs among query complexity, additional gate complexity, ancilla count, and normalization factor. 
Importantly, complete factorization into linear factors yields our logarithmic-depth construction; in the single-oracle model, it uses $2^d-1$ queries to the input block encoding and saturates the corresponding query lower bound.
Moreover, we derive error bounds for the polynomial approximation and the implementation of block encoding. We demonstrate the framework on nonlinear activation functions and pixel-wise image transformations. 
Furthermore, our framework is implemented as explicit circuits in Python; the code is available at \url{https://github.com/ChunlinYANG0/QuantumHadamardProcedure}.

The remainder of the paper is organized as follows. 
Section~\ref{sec: preliminaries} formulates the problem, defines the oracle models and resource metrics, and introduces the necessary building blocks. 
Section~\ref{sec: QHMF} presents the factorization, binary-tree, and trade-off constructions and analyzes their resource costs. 
Section~\ref{sec: error} provides the error analysis. 
Section~\ref{sec: applications} demonstrates the applications. Section~\ref{sec: conclusion} concludes with a summary and outlook.

\section{Problem formulation}\label{sec: preliminaries}

\subsection{Problem Setting}
The central task of this work is nonlinear amplitude transformation: given an $n$-qubit matrix $A = \sum_{ij} a_{ij}\ket{i}\bra{j} \in \mathbb{C}^{2^n\times2^n}$ with all entries on $\Omega\subseteq\mathbb{C}$ and a target function $f:\Omega\to\mathbb{C}$ that can be approximated by a polynomial, produce a new matrix whose entries are $f(a_{ij})$, i.e., $f(A)=\left(f(a_{ij})\right)$. 

The first step to implement the nonlinear amplitude transformation is the polynomial approximation. 
By Mergelyan's theorem~\cite{mergelyan1951,mergelyan1952}, if $\Omega$ is compact, $\mathbb{C}\setminus \Omega$ is connected, and $f: \Omega\to\mathbb{C}$ is continuous on $\Omega$ and holomorphic in the interior of $\Omega$, then $f$ can be uniformly approximated on $\Omega$ by polynomials.
Therefore, for a prescribed accuracy, the task reduces to implementing a polynomial entry-wise on $A$, given by the following definition.
\begin{definition}[Hadamard Matrix Function~\cite{horn1994topics}] 
	Let $A=\left(a_{ij}\right)$ be an $n$-qubit matrix. 
    A Hadamard matrix function of degree $2^d-1$ acting on $A$ is 
	\begin{equation}\label{eq: Hadamard matrix function}
		P \left(A\right) = \sum_{k=0}^{2^d-1} c_k A^{\circ k} = \left(\sum_{k=0}^{2^d-1}c_{k}a_{ij}^{k}\right),
	\end{equation}
	where $A^{\circ k}=\left(a_{ij}^k\right)$ denotes the $k$-fold Hadamard (entry-wise) power of $A$, and $A^{\circ 0}=J$ is the all-ones matrix.
\end{definition}
Here, the degree is chosen as $2^d-1$ so that the number of terms is a power of two. For convenience, the quantum realizations of the Hadamard product and the Hadamard matrix function are called QHP and QHMF, respectively.

Approximating $f$ by $P$ replaces the target transformation with a Hadamard matrix function of finite degree. This replacement introduces a scalar approximation error on $\Omega$, while implementing $P(A)$ through quantum circuits introduces an implementation error; the two together form the total error budget that any admissible construction must meet.
These quantities, together with the post-selection success probability $p_{\rm succ}$, are defined as metrics in Sec.~\ref{subsec: resource metrics}, and their propagation through the constructions of Sec.~\ref{sec: QHMF} is analyzed in Sec.~\ref{sec: error}.

\subsection{Oracle Model}\label{subsec: oracle models}
We adopt the block-encoding input model~\cite{gilyen2019quantum}, in which a non-unitary matrix is represented as a submatrix of a unitary operator. Several block-encoding protocols have been proposed~\cite{camps2022fable,camps2024explicit,sunderhauf2024block,yang2025dictionary,li2025binary} and can be used to realize this input model.
\begin{definition}[Block Encoding~\cite{gilyen2019quantum}]
	Suppose that $A$ is an $n$-qubit matrix, $\alpha,\epsilon\in\mathbb{R}_{+}$ and $a \in \mathbb{N}$, then we say that the $(n + a)$-qubit unitary $U_{A}$ is an $(\alpha,a,\epsilon )$-block-encoding of $A$, if
	\begin{equation*}
		\left\|A-\alpha\left(\bra{0}^{\otimes a}\otimes I_{2^n} \right) U_{A} \left(\ket{0}^{\otimes a} \otimes I_{2^n} \right)   \right\| \leq \epsilon,
	\end{equation*}
    where $\alpha$ is the normalization factor satisfying $\alpha\geq\|A\|$.
\end{definition}

In this paper, QHMF operates on an $n$-qubit matrix $A$, and access to $A$ is through oracles. We present two oracle models as follows. 
\begin{itemize}
    \item \textit{Single-oracle model.} 
    Only an $(\alpha_0,a_0,\epsilon_0)$-block-encoding $U_A$ of $A$, together with its controlled versions, is available. 
    No block encodings of higher Hadamard powers $A^{\circ 2^l}$ for $l\geq1$ are assumed; if needed, they must be synthesized from $U_A$ via the quantum Hadamard product.

    \item \textit{Power-oracle model.} 
    $(\alpha_l,a_l,\epsilon_l)$-block-encodings $U_{A^{\circ 2^l}}$ of $A^{\circ 2^l}$ for $l\in\{0,1,\cdots,d-1\}$, together with their controlled versions, are supplied as primitive oracles. These block encodings are treated as primitive oracles. 
\end{itemize}

In both models, the all-ones matrix $J$ is not treated as an oracle. 
Its exact block encoding is implemented explicitly, and its cost is included in the non-oracle gate and ancilla counts. 
This convention ensures that resource comparisons between different constructions remain fair. 
Moreover, query complexities are only comparable within the same oracle model; comparing across models requires including the cost of synthesizing the power oracles from $U_A$.

\subsection{Resource Metrics}\label{subsec: resource metrics}
To compare different QHMF constructions, we consider the following metrics. All counts are for the additional circuit beyond the oracle implementations unless stated otherwise. Throughout, $U_{P(A)}$ denotes the unitary implemented by a QHMF
circuit with $N$ ancilla qubits, 
\begin{equation*}
    \overline{P}(A) = \left(\bra{0}^{\otimes N}\otimes I_{2^n}\right) U_{P(A)} \left(\ket{0}^{\otimes N}\otimes I_{2^n}\right)
\end{equation*}
denotes the matrix it block-encodes, and $\|\cdot\|$, $\|\cdot\|_{\max}$ and $\|\cdot\|_F$ denote the spectral, entry-wise maximum and Frobenius norms, respectively.
\begin{itemize}
    \item Query complexity $Q_q^l$.\\
    Let $Q_q^l$ be the query time to the $m$-qubit controlled block encoding $U_{A^{\circ 2^l}}$ of $A^{\circ 2^l}$ for $l,q\geq0$. 
    In the power-oracle model, we use $\left\{Q_q^l:q\geq0\right\}_{l=0}^{d-1}$. 
    In the single-oracle model, only $Q^0_q$ for $q\geq0$ are used. 
    Query complexities are only comparable within the same oracle model.

    \item Additional depth $D$ and size $S$.\\
    Let $D$ be the circuit depth and $S$ the total number of elementary gates in $\left\{\operatorname{U}(2), \operatorname{CNOT}\right\}$, excluding the internal cost of the oracle implementations. 
    The all-ones matrix $J$ is implemented explicitly, and its depth and size are included in $D$ and $S$.

    \item Ancilla count $N$.\\
    Let $N$ be the number of ancilla qubits used by the circuit, excluding the data register. 
    Borrowed ancillas, which are returned to their initial state and can therefore be reused, are counted once by their maximum simultaneous usage.

    \item Normalization factor $\alpha$.\\
    The circuit block-encodes $P(A)$ with normalization factor $\alpha$, i.e., $P(A)/\alpha$ is the submatrix implemented by the circuit. 
    A larger $\alpha$ reduces the post-selection success probability and increases the absolute implementation error, but does not necessarily degrade the relative error.

    \item Post-selection success probability $p_{\rm succ}$. \\
    For a uniformly distributed input state $\rho=I_{2^n}/2^n$, the average probability of projecting the ancilla onto $\ket{0}$ after applying the circuit is
    \begin{align*}
        p_{\rm succ} =& \operatorname{Tr}\left(\rho \left(\frac{P(A)}{\alpha}\right)^\dagger \left(\frac{P(A)}{\alpha}\right)\right) \\
        =& \frac{\|P(A)\|_F^2}{\alpha^2 2^n},
    \end{align*}
    where $\|\cdot\|_F$ is the Frobenius norm. 
    This quantity determines the sampling cost: the expected number of circuit repetitions is $\mathcal{O}(1/p_{\rm succ})$.

    \item Total error budget $\varepsilon$. \\
    Two errors enter the total budget. 
    The \emph{approximation error} $\epsilon_{\rm appro}$ arises from replacing the target function $f$ by the polynomial $P$,
    \begin{equation*}
        \epsilon_{\rm appro} = \max_{z\in\Omega}|f(z)-P(z)|.
    \end{equation*}
    Because $P$ acts entry-wise on $A$, this scalar error propagates to the matrix elements without amplification, 
    \begin{equation*}
        \|f(A)-P(A)\|_{\max}\leq\epsilon_{\rm appro}.
    \end{equation*}
    The \emph{implementation error} $\epsilon_{\rm impl}$ is the error of the circuit in the block-encoding convention. It is quantified in the spectral norm at the two scales
    \begin{gather*}
        \epsilon_{\rm abs} = \epsilon_{\rm impl} = \left\|P(A)-\alpha\overline{P}(A)\right\|,
        \\
        \epsilon_{\rm rel} = \frac{\epsilon_{\rm abs}}{\alpha} = \left\|\frac{P(A)}{\alpha}-\overline{P}(A)\right\|.
    \end{gather*} 
    The absolute error $\epsilon_{\rm abs}$ is the quantity entering the block-encoding convention and the composition rules, whereas its relative error $\epsilon_{\rm rel}$ is the accuracy of the post-selected output, since $\overline{P}(A)$ replaces the ideal $P(A)/\alpha$ on the post-selected branch; the bounds on $\epsilon_{\rm rel}$ derived in Sec.~\ref{sec: error} are independent of $\alpha$, and it is therefore the measure in which constructions with different normalization factors are compared.
    Since $\|\cdot\|_{\max}\leq\|\cdot\|$, the two errors combine into the entry-wise total error
    \begin{equation*}
        \left\|f(A)-\alpha\overline{P}(A)\right\|_{\max}
        \leq\epsilon_{\rm appro}+\epsilon_{\rm impl} \leq \varepsilon,
    \end{equation*}
    which is the accuracy criterion adopted in this work.
\end{itemize}

\subsection{Building Block}
\subsubsection{Quantum Hadamard Product}
QHP is the fundamental operation for amplitude transformations through quantum circuits. An efficient implementation for two matrices was given in~\cite{zhao2021compiling,guo2025quantum}.
\begin{lemma}[QHP of Two Matrices~\cite{zhao2021compiling,guo2025quantum}]
    Assume access to an $\left(\alpha,a,\epsilon\right)$-block-encoding of an $n$-qubit matrix $A$ and a $\left(\beta,b,\delta\right)$-block-encoding of an $n$-qubit matrix $B$. Then the circuit in Fig.~\ref{circuit: Hadamard product of two matrices} implements a $\left(\alpha\beta,a+b+n,\alpha\delta+\beta\epsilon\right)$-block-encoding of $A \circ B$. 
	\begin{figure}[htbp]
		\centering
		\begin{tikzpicture}
			\begin{yquant}
				qubit {} b;
				qubit {} jb;
				qubit {} a;
				qubit {$\ket{j}$} ja;
				
				["north:$n$" 
				{font=\protect\footnotesize, inner sep=0pt}]
				slash ja,jb;
				["north:$b$" 
				{font=\protect\footnotesize, inner sep=0pt}]
				slash b;
				["north:$a$" 
				{font=\protect\footnotesize, inner sep=0pt}]
				slash a;
				hspace {5pt} -;
				cnot jb | ja;
				hspace {5pt} -;
				box {$U_B$} (b,jb);
				box {$U_A$} (a,ja);
				hspace {5pt} -;
				cnot jb | ja;
				hspace {5pt} -;
				
				output {$\ket{i}$} ja;
			\end{yquant}
		\end{tikzpicture}
		\caption{Quantum circuit for the Hadamard product of two $n$-qubit matrices.}
		\label{circuit: Hadamard product of two matrices}
	\end{figure}
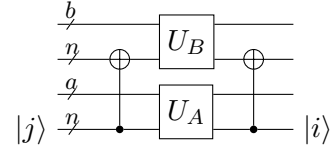
\end{lemma}

By iterating this construction, we obtain the Hadamard product of multiple matrices. Using a multi‑register GHZ‑type operator (denoted $GHZ_m^n$) of depth $\mathcal{O}(\log m)$ and size $\mathcal{O}\left(nm\right)$~\cite{watts2019exponential},  the following corollary holds; proof provided in Appendix~\ref{sec: proof of corollary QHP of m matrices}.
\begin{corollary}[QHP of Multiple Matrices]\label{corollary: QHP of m matrices}
	Assume access to $\left(\tilde{\alpha}_k,\tilde{a}_k,\tilde{\epsilon}_k\right)$-block-encoding $U_{A_k}$ of $n$-qubit matrix $A_k=\left(a^{(k)}_{ij}\right)\in\mathbb{C}^{2^n\times 2^n}$ for $k\in\{0,1,\cdots,m-1\}$. Then the unitary represented by the circuit shown in Fig.~\ref{circuit: QHP of d matrices} is a $\left(\alpha,a,\epsilon\right)$-block-encoding of $\bigcirc_{k=0}^{m-1} A_{k}$, where $\alpha=\prod_{k=0}^{m-1}\tilde{\alpha}_k$, $a=(m-1)n+\sum_{k=0}^{m-1}\tilde{a}_k$, and $\epsilon=\left(\sum_{k=0}^{m-1} \frac{\tilde{\epsilon}_k}{\tilde{\alpha}_k}\right) \prod_{k=0}^{m-1}\tilde{\alpha}_{k}$.
\end{corollary}
\begin{figure}[htbp]
    \centering
    \includegraphics[width=0.75\linewidth]{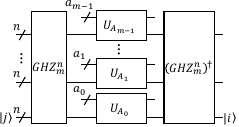}
    \caption{Quantum circuit for the Hadamard product of $d$ $n$-qubit matrices.}
    \label{circuit: QHP of d matrices}
\end{figure}

All block encodings in Corollary~\ref{corollary: QHP of m matrices} are applied in parallel. If the entire circuit is controlled by additional qubits, this parallelism will be lost and the depth will become additive.

\subsubsection{Block Encoding of All-Ones Matrix}\label{subsec: block encoding of all-ones matrix}
The constant term of a Hadamard matrix function involves the all-ones matrix $J$. The following elementary LCU construction realizes its block encoding with the optimal complexity.

\begin{lemma}[Constant-Depth Block Encoding of All-Ones Matrix]\label{lem: block encoding of all-ones matrix}
    Let $J\in\mathbb{C}^{2^n\times2^n}$ be an all-ones matrix. The circuit in Fig.~\ref{circuit: block encoding of all-ones matrix} implements a $\left(2^n,n,0\right)$-block-encoding of $J$. 
\end{lemma}
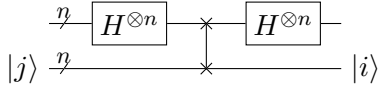
\begin{figure}[htbp]
    \centering
    \begin{tikzpicture}
        \begin{yquant}
            qubit {} anc;
            qubit {$\ket{j}$} j;
            
            ["north:$n$" 
            {font=\protect\footnotesize, inner sep=0pt}]
            slash anc,j;

            hspace {5pt} -;
            box {$H^{\otimes n}$} anc;
            hspace {5pt} -;
            swap (anc,j);
            hspace {5pt} -;
            box {$H^{\otimes n}$} anc;
            hspace {5pt} -;
            
            output {$\ket{i}$} j;
        \end{yquant}
    \end{tikzpicture}
    \caption{Quantum circuit for block encoding of all-ones matrix.}
    \label{circuit: block encoding of all-ones matrix}
\end{figure}
\begin{proof}
    The encoded matrix can be computed as follows:
    \begin{equation*}
        \begin{aligned}
            &\left(\bra{0}^{\otimes n}\bra{i}\right) \left(H^{\otimes n}\otimes I_{2^n}\right) \operatorname{SWAP} \\
            &\cdot \left(H^{\otimes n}\otimes I_{2^n}\right) \left(\ket{0}^{\otimes n}\ket{j}\right) \\
            =& \frac{1}{2^n} \left(\sum_{l=0}^{2^n-1}\bra{l}\bra{i}\right) \left(\sum_{l=0}^{2^n-1}\ket{j}\ket{l}\right) \\
            =& \frac{1}{2^n}.
        \end{aligned}  
    \end{equation*}
\end{proof}

The block encoding of Lemma~\ref{lem: block encoding of all-ones matrix} is optimal in the following rigorous sense:
\begin{itemize}
    \item Normalization factor: 
    $\alpha = 2^n$ is optimal since $\|J\|_2 = 2^n$.
    
    \item Encoding error: 
    The construction is exact ($\epsilon=0$), which is clearly optimal.
    
    \item Circuit depth: 
    The construction uses $\Omega(1)$ depth, which is constant.
    
    \item Circuit size: 
    The circuit uses $\Theta(n)$ elementary gates. 
    To see that this is asymptotically optimal, observe that every data qubit must be acted upon by at least one gate.
    If some data qubit $q$ is never touched, then the encoded matrix element $\left(\bra{0}^{\otimes a}\bra{i}\right) U_J \left(\ket{0}^{\otimes a}\ket{j}\right)$ necessarily contains a factor $\delta_{i_q,j_q}$.
    This contradicts the requirement that every entry of $J/2^n$ equals $1/2^n$, independent of both $i_q$ and $j_q$.
    In the standard gate set $\{U(2),\mathrm{CNOT}\}$, each gate acts on at most two data qubits.
    Consequently at least $\Omega(n)$ gates are required.
    The present construction matches this lower bound.
\end{itemize}
Therefore, Lemma~\ref{lem: block encoding of all-ones matrix} achieves the optimal trade-off among normalization, depth, size, and ancilla count within the stated resource model.

\section{Quantum Hadamard Matrix Function}\label{sec: QHMF}
Our core algorithmic contributions in this work exploit the fundamental theorem of algebra to achieve substantially improved resource efficiency. We begin with the most depth-efficient construction, factorization via linear factors, which achieves logarithmic circuit depth at the cost of exponential ancilla and a large normalization factor. To mitigate this overhead, we introduce a tunable framework that decomposes the polynomial into low-degree factors and then develop a binary-tree-based subroutine that serves as the building block for implementing the factors, offering a favorable middle ground in the depth–ancilla trade-off.

\subsection{Logarithmic-depth Factorization Construction}\label{sec: factorization-based QHMF}
We begin by exploring the fundamental limit of circuit depth for implementing a Hadamard matrix function. The key observation is that by the fundamental theorem of algebra, any polynomial can be decomposed into linear factors over the complex numbers. That is, a degree-$(2^d-1)$ polynomial $P(x)=\sum_{k=0}^{2^d-1} c_k x^k$ with $c_{2^d-1}\neq 0$ can be written as
\begin{equation*}
    P(x) = c_{2^d-1} \prod_{k=0}^{2^d-2} (x - r_k),
\end{equation*}
where $\{r_k\}_{k=0}^{2^d-2}$ are the roots of $P$. Since the Hadamard product distributes over addition and scalar multiplication, the corresponding matrix polynomial factorizes as
\begin{equation*}
    P(A) = c_{2^d-1} \bigcirc_{k=0}^{2^d-2} (A - r_k J).
\end{equation*}
This factorization reduces QHMF to the implementation of a sequence of linear factors $A - r_k J$, each of which is a simple linear combination of $A$ and $J$. We assume the roots are provided classically~\cite{kerner1966ein, aberth1973iteration, jenkins1970three, horn1999matrix}; for high-degree polynomials, their computation may itself be non-trivial but is one-time.

Since all linear factors are independent up to the Hadamard product, we can implement them in parallel. The $2^d-1$ factors are encoded on $2^d-1$ separate register pairs, each consisting of one index qubit, the $\hat a_0=\max\{a_0,n\}$ ancilla and the $n$ data qubits for $U_A,U_J$. Finally, we apply the QHP to combine them. The resulting overall circuit is depicted in Fig.~\ref{circuit: factorization-based QHMF}.
\begin{figure*}[htbp]
    \centering
    \includegraphics[width=0.8\linewidth]{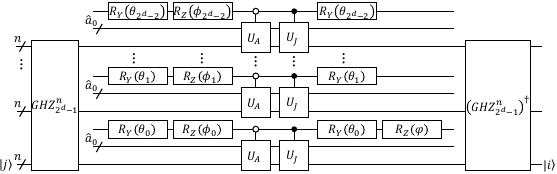}
    \caption{Quantum circuit for the factorization-based QHMF, where $\hat{a}_0=\max\{a_0,n\}$.}
    \label{circuit: factorization-based QHMF}
\end{figure*}

We formalize this construction in the following theorem.
\begin{theorem}[Factorization-based QHMF]\label{thm: factorization-based QHMF}
    Let $P(x)=\sum_{k=0}^{2^d-1}c_kx^k$ be a $\left(2^d-1\right)$-degree polynomial with roots $\left\{r_k\right\}_{k=0}^{2^d-2}$. In the single-oracle model, assume access to an $(\alpha_0,a_0,\epsilon_0)$-block-encoding $U_A$ of $A\in\mathbb{C}^{2^n\times2^n}$ together with its controlled version. No block encodings of higher Hadamard powers $A^{\circ 2^l}$ for $l\ge1$ are assumed. Then there exist angle sequences
    \begin{gather*}
        \left\{\theta_k = 2\arctan\sqrt{\frac{2^n\left|r_k\right|}{\alpha_0}}\right\}_{k=0}^{2^d-2}, \\
        \left\{\phi_k = \operatorname{Arg}\left(r_k\right)\right\}_{k=0}^{2^d-2}, \\
        \varphi = -\sum_{k=0}^{2^d-2}\operatorname{Arg}\left(r_k\right) - 2\operatorname{Arg}\left(c_{2^d-1}\right)
    \end{gather*}
    such that the circuit in Fig.~\ref{circuit: factorization-based QHMF} implements the Hadamard matrix function $P(A)$ with a normalization factor
    \begin{equation*}
        \alpha_{\rm Fac} =\left|c_{2^d-1}\right|\prod_{k=0}^{2^d-2}\left(\alpha_0+2^n\left|r_k\right|\right).
    \end{equation*}
\end{theorem}

The proof of Theorem~\ref{thm: factorization-based QHMF} follows from a direct computation of the block-encoded matrix for each factor and multiplying the results via the Hadamard product. A detailed derivation is provided in Appendix~\ref{sec: proof of factorization-based QHMF}. 

We note that the normalization factor of factorization-based QHMF is large. Suppose that $\alpha_0 = \mathcal{O}\left(2^n\right)$, we have $\alpha_{\rm Fac} = \mathcal{O}\left(2^{n2^d}\right)$. This is an inherent cost of the factorization approach. For larger $n$ and $d$, this is unacceptable since it affects the success probability of post-selection. In the next section, we provide the trade-off construction that can reduce it.

We now analyze the resource costs of this construction. The circuit consists of three types of operations:
\begin{enumerate}
    \item GHZ-type operators $\operatorname{GHZ}_{2^d-1}^n$. Each consists of $n$ generators of $(2^d-1)$-qubit GHZ states, where each generator contains $2^d-2$ CNOT gates. These can be implemented with depth $\Omega(d)$ and size $\Theta(n 2^d)$ using standard binary tree constructions~\cite{watts2019exponential}.
    \item Single-qubit rotations $R_Y(\theta_k)$, $R_Z(\phi_k)$, and a final $R_Z^\dagger(\varphi)$. They contribute $\mathcal{O}(1)$ to the depth and $\mathcal{O}(2^d)$ to the size.
    \item Controlled block-encodings $U_A$ and $U_J$ for each factor. Since they are both called once per factor, the total number of queries to $U_A$ is $2^d-1$, all in parallel. The additional depth and size are dominated by those of a single controlled $U_J$, which are both $\Theta(n)$.
\end{enumerate}

Summing up, we obtain the following complexity bounds.
\begin{corollary}[Complexity of Factorization-based QHMF]\label{cor: complexity of factorization-based QHMF}
	The factorization-based QHMF in Theorem~\ref{thm: factorization-based QHMF} can be implemented with:
    \begin{itemize}
        \item Query complexity: $Q_1^0=\mathcal{O}\left(2^d\right)$ in parallel.
        \item Additional gates: $D=\mathcal{O}\left(n+d\right)$,  $S=\mathcal{O}\left(n2^d\right)$.
        \item Ancilla count: $N=\mathcal{O}\left(2^d(n+a_0)\right)$.
    \end{itemize}
\end{corollary}
\begin{proof}
    The result is obviously obtained from the circuit in Fig.~\ref{circuit: factorization-based QHMF}.
\end{proof}

We now present the lower bound on query complexity for implementing the Hadamard matrix function provided only in the single-oracle model.
\begin{proposition}[Query Lower Bound in Single-Oracle Model]\label{pro: query lower bound}
    Let $P(x)=\sum_{k=0}^{2^d-1} c_k x^k$ be a $(2^d-1)$-degree polynomial with $c_{2^d-1}\neq0$, and $A$ be an $n$-qubit matrix. In the single-oracle model, assume that the only $A$-dependent operation available is a block encoding $U_A$ of $A$, together with its controlled version. Then any quantum circuit that implements the Hadamard matrix function $P(A)$ must make at least $\Omega(2^d)$ queries to $U_A$.
\end{proposition}
\begin{proof}
    Fix $i,j$ and consider the matrix family $A(t)=A_0+tE_{ij}$. 
    After $q$ queries to $U_{A(t)}$, the final amplitude is a polynomial in $t$ of degree at most $q$.
    This follows by induction on $q$: each query inserts the encoded matrix $A(t)$, which is linear in $t$, into the amplitude, and all other operations are $t$-independent. 
    Post-selecting the ancilla on $\ket{0}$ yields the $(i,j)$ matrix element, which is therefore a polynomial in $t$ of degree at most $q$.
    On the other hand, the target Hadamard matrix function satisfies $\bra{i}P(A(t))\ket{j} = P(t)$, which is a polynomial in $t$ of degree $2^d-1$ with nonzero leading coefficient. 
    Hence, we must have $q\ge 2^d-1$.
\end{proof}

Since the factorization-based construction uses exactly $2^d-1$ queries to $U_A$, it saturates the query lower bound of Proposition~\ref{pro: query lower bound} in the single-oracle model.
Moreover, its additional depth $\mathcal{O}(n+d)$ is independent of the system size $n$ and grows only logarithmically with the polynomial degree.

However, this extreme depth performance comes at a steep price: the ancilla count and the normalization factor grow exponentially with $d$. For near-term quantum devices with limited qubit counts, this exponential overhead is prohibitive. This motivates us to seek a tunable balance between depth and width, which we develop in the next section by decomposing the polynomial into a product of lower-degree factors and implementing each factor with a more efficient subroutine.

\subsection{Depth-Ancilla-Tunable Construction}\label{sec: trade-off}
The extreme parallelism of the factorization method implies that circuit depth can be reduced to a logarithmic depth in $d$ at the cost of an exponential number of ancilla and a large normalization factor. This extreme is often impractical for near-term devices. Conversely, if we wish to minimize them, we must reduce parallelism and reuse qubits, inevitably increasing the depth and size. This motivates the following question: \emph{can depth and size be traded against width and normalization factor in a controllable way?}

We answer this affirmatively by decomposing the target polynomial into a Hadamard product of multiple lower-degree factors. Let $P(x)$ be a polynomial of degree $2^d-1$. For any integer $m$ with $1 \leq m \leq 2^d-1$, we factor $P$ as
\begin{equation*}
    P(x) = \prod_{s=0}^{m-1} P^{(s)}(x),
\end{equation*}
where each $P^{(s)}$ is a polynomial with $1\leq\deg P^{(s)}\leq K = \left\lceil \frac{2^d-1}{m} \right\rceil$. Such a factorization always exists by grouping the roots of $P$ into $m$ disjoint subsets $\{\mathcal R_s\}_{s=0}^{m-1}$ and defining $P^{(s)}(x) = c^{(s)} \prod_{r_{k_s} \in \mathcal R_s} (x - r_{k_s})$, with $\prod_s c^{(s)} = c_{2^d-1}$. The grouping is arbitrary; for resource optimization, we may choose balanced degrees to minimize the maximum degree.

Leveraging the distributivity of the Hadamard product, $P(A)$ factorizes accordingly:
\begin{equation*}
    P(A) = \bigcirc_{s=0}^{m-1} P^{(s)}(A).
\end{equation*}
Thus, the overall implementation reduces to (i) realizing each low-degree Hadamard matrix function $P^{(s)}(A)$ independently with low ancilla overhead, and (ii) combining them via the QHP. This decomposition provides a flexible knob: by adjusting $m$, we can interpolate between the two extremes:
\begin{itemize}
    \item $m = 2^d-1$: each $P^{(k)}$ is linear. This recovers the factorization method of Sec.~\ref{sec: factorization-based QHMF} with low depth and high ancilla.
    \item $m = 1$: a single high-degree polynomial. This reduces to the binary-tree method we will present in Sec.~\ref{sec: binary-tree-based QHMF} with high depth and low ancilla.
    \item Intermediate $m$: a tunable trade-off, allowing the user to choose a point that best fits the available qubit count and coherence time.
\end{itemize}

The overall circuit architecture is shown in Fig.~\ref{circuit: trade-off}. The $m$ factors are implemented in parallel on separate registers; each factor $P^{(k)}(A)$ is implemented by a subroutine, and the results are combined via the QHP. 

\begin{figure}[htbp]
	\centering
	\includegraphics[width=0.75\linewidth]{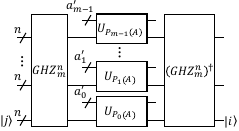}
	\caption{General trade-off circuit for implementing $P(A) = \bigcirc_{s=0}^{m-1} P^{(s)}(A)$. Each $P^{(s)}(A)$ is implemented by a low-degree Hadamard-polynomial circuit. The parallel Hadamard product combines the results.}
	\label{circuit: trade-off}
\end{figure}

We now derive the resource costs of this generic framework as functions of $m$, assuming that each $P^{(s)}(A)$ is implemented by a subroutine with known complexity. Let $\mathcal{Q}_q^l(K)$, $\mathcal{D}(K)$, $\mathcal{S}(K)$, and $\mathcal{N}(K)$ denote the query to $q$-qubit controlled $U_{A^{\circ2^l}}$, additional depth, additional size, and ancilla count of a subroutine that implements a Hadamard matrix function of degree $K$, where $l\in\{0,1,\cdots,\left\lceil\log K\right\rceil-1\}$. Let $\alpha(P(A))$ be the normalization factor of a Hadamard matrix function $P(A)$. Then the overall circuit in Fig.~\ref{circuit: trade-off} has:
\begin{itemize}
    \item Query complexity: $Q_q^l(m) = m \mathcal{Q}_q^l(K)$ for $q\geq0$.
    \item Additional depth: $D(m) = \mathcal{D}(K) + \mathcal{O}(\log m)$.
    \item Additional size: $S(m) = m \mathcal{S}(K) + \mathcal{O}(n m)$.
    \item Ancilla count: $N(m) = m \mathcal{N}(K) + \mathcal{O}(n m)$.
    \item Normalization factor: $\alpha_{\rm TO}(m) = \prod_{s=0}^{m-1} \alpha\left(P^{(s)}(A)\right)$.
\end{itemize}
 
This framework is summarized in the following theorem.
\begin{theorem}[Depth-Ancilla Trade-Off QHMF]\label{thm: trade-off QHMF}
    Let $m$ be an integer with $1\leq m \leq 2^d-1$. Set $K = \left\lceil \frac{2^d-1}{m} \right\rceil$. Let $P(x)=\prod_{s=0}^{m-1}P^{(s)}(x)$ be a $\left(2^d-1\right)$-degree polynomial, where each $P^{(s)}$ has degree no more than $K$. In the power-oracle model, assume access to $\left(\alpha_l, a_l, \epsilon_l\right)$-block-encodings of $A^{\circ 2^l}$ for $l\in\left\{0,1,\dots,\left\lceil\log K\right\rceil-1\right\}$ together with their controlled versions. Then the circuit in Fig~\ref{circuit: trade-off} implements the Hadamard matrix function $P(A)$ with:
    \begin{align*}
        Q_q^l\left(m\right) &= m \mathcal{Q}_q^l\left(\left\lceil\frac{2^d-1}{m}\right\rceil\right) \mbox{ for } q\geq0, \\
        D\left(m\right) &= \mathcal{D}\left(\left\lceil\frac{2^d-1}{m}\right\rceil\right) + \mathcal{O}(\log m), \\
        S\left(m\right) &= m \mathcal{S}\left(\left\lceil\frac{2^d-1}{m}\right\rceil\right) + \mathcal{O}(n m), \\
        N\left(m\right) &= m \mathcal{N}\left(\left\lceil\frac{2^d-1}{m}\right\rceil\right) + \mathcal{O}(n m), \\
        \alpha_{\rm TO}(m) &= \prod_{s=0}^{m-1} \alpha\left(P^{(s)}(A)\right).
    \end{align*}
\end{theorem}
\begin{proof}
    By substitude $K = \left\lceil \frac{2^d-1}{m} \right\rceil$ into $Q_q^l\left(m\right)$, $D\left(m\right)$, $S\left(m\right)$, and $N\left(m\right)$, the result holds.
\end{proof}

The trade-off framework is the central architectural contribution of this work. It decouples the complexity of the polynomial degree from the hardware constraints, allowing the user to select $m$ based on the available resources.

\subsection{Binary-tree Construction for Low-Degree Factors}\label{sec: binary-tree-based QHMF}
The trade-off framework reduces the problem to implementing a collection of low-degree Hadamard matrix functions. For these factors, we require a subroutine that strikes a balance between depth and ancilla, avoiding the exponential ancilla overhead of the full factorization while still offering better performance in terms of depth. We now present a Horner-inspired binary-tree construction that achieves exactly this balance.

\begin{figure*}[htbp]
	\centering
	\includegraphics[width=0.85\textwidth]{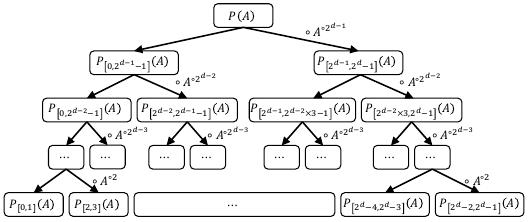}
	\caption{Binary-tree division of the Hadamard matrix function $P\left(A\right)$.}
	\label{fig: binary-tree division}
\end{figure*}

The key idea is to evaluate the polynomial in a hierarchical manner, reusing intermediate Hadamard powers, and thereby reducing queries to higher-order powers. This is analogous to Horner's rule for scalar polynomials, but adapted to the Hadamard-product setting.

For a $\left(2^d-1\right)$-degree Hadamard matrix function $P\left(A\right)=\sum_{k=0}^{2^d-1} c_k A^{\circ k}$ and $i,j \in \{0,1,\cdots,2^d-1\}$ with $i<j$, where $c_k\in\mathbb{C}$, we denote $P_{[i,j]}\left(A\right) = \sum_{k=i}^j c_k A^{\circ (k-i)}$.

Firstly, we divide $P\left(A\right)$ into two parts,
\begin{equation*}
	\begin{aligned}
		P\left(A\right) 
		= P_{\left[0, 2^{d-1}-1\right]}\left(A\right) + P_{\left[2^{d-1}, 2^{d}-1\right]}\left(A\right)\circ A^{\circ 2^{d-1}},
	\end{aligned}
\end{equation*}
where $P_{\left[0, 2^{d-1}-1\right]}\left(A\right)$ and $P_{\left[2^{d-1}, 2^{d}-1\right]}\left(A\right)$ have degrees no more than $2^{d-1}-1$. Then, dividing $P_{\left[0, 2^{d-1}-1\right]}\left(A\right)$ and $P_{\left[2^{d-1}, 2^{d}-1\right]}\left(A\right)$ into two parts, respectively, yields that
\begin{equation*}
	\begin{aligned}
	    P_{\left[0, 2^{d-1}-1\right]}\left(A\right)
		=& P_{\left[0, 2^{d-2}-1\right]}\left(A\right) \\
        &+ P_{\left[2^{d-2}, 2^{d-1}-1\right]}\left(A\right) \circ A^{\circ 2^{d-2}}, 
	\end{aligned}
\end{equation*}
and
\begin{equation*}
	\begin{aligned}
		P_{\left[2^{d-1}, 2^{d}-1\right]}\left(A\right) 
		=& P_{\left[2^{d-1}, 2^{d-2}\times3-1\right]}\left(A\right) \\
        &+ P_{\left[2^{d-2}\times3, 2^{d}-1\right]}\left(A\right) \circ A^{\circ 2^{d-2}},
	\end{aligned}
\end{equation*}
where the resulting $P_{\left[0, 2^{d-2}-1\right]}\left(A\right)$, $P_{\left[2^{d-2}, 2^{d-1}-1\right]}\left(A\right)$, $P_{\left[2^{d-1}, 2^{d-2}\times3-1\right]}\left(A\right)$, $P_{\left[2^{d-2}\times3, 2^{d}-1\right]}\left(A\right)$ have degrees no more than $2^{d-2}-1$. The division is recursively applied until the degrees of the obtained Hadamard matrix functions are equal to $1$. 
This process can be described by a binary tree shown in Fig.~\ref{fig: binary-tree division}, where the Hadamard matrix functions of the leaf nodes have the form of
\begin{equation}\label{eq: rule of binary tree for leaf node}
	P_{\left[2k_{d-1},2k_{d-1}+1\right]}\left(A\right) = c_{2k_{d-1}} J + c_{2k_{d-1}+1} A,
\end{equation}
for $k_{d-1}\in\{0,1,\cdots,2^{d-1}-1\}$.

The key advantage of this hierarchical splitting is that it requires only block-encodings of the powers $A^{\circ 2^0}, A^{\circ 2^1}, \dots, A^{\circ 2^{d-1}}$, i.e., only $d$ distinct matrices. The binary-tree construction thus reduces the ancilla overhead, as intermediate results are reused.

For $l\in\left\{0,1,\cdots,d-2\right\}$ and $k_l\in\left\{0,1,\cdots,2^l-1\right\}$, define the interval normalization factors
\begin{equation}\label{eq: interval normalization factor}
    \alpha_{l, k_l} = \sum_{k=2^{d-l}k_l}^{2^{d-l}(k_l+1)-1} 2^{n\overline{\operatorname{bit}_0(k)}}\alpha_0^{\operatorname{bit}_0(k)} \left|c_k\right|\prod_{j=1}^{d-l-1}\alpha_j^{\operatorname{bit}_j(k)},
\end{equation}
where $\operatorname{bit}_j(k)$ denotes the $j$-th binary digit of $k$ and the overline is the NOT operation. 

\begin{figure*}[htbp]
	\centering
	\begin{minipage}[b]{\textwidth}
		\centering
		\resizebox{\textwidth}{!}{
			\begin{tikzpicture}
				\begin{yquant}
					qubit {} add;
					qubit {} uaanc;
					qubit {} uaj;
					qubit {} anc;
					qubit {$\ket{j}$} j;
					
					["north:$a_{d-l-1}$" 
					{font=\protect\footnotesize, inner sep=0pt}]
					slash uaanc;
					["north:$n$" 
					{font=\protect\footnotesize, inner sep=0pt}]
					slash uaj, j;
					["north:$b_{l}$" 
					{font=\protect\footnotesize, inner sep=0pt}]
					slash anc;
					
					box {$U_{P_{\left[2^{d-l}k_l,2^{d-l}(k_l+1)-1\right]}\left(A\right)}$} (add,uaanc,uaj,anc,j);
					
					text {$=$} (-);
					
					box {$R_Y\left(\gamma_{l,k_l}\right)$} add;
					box {$U_{P_{\left[2^{d-l}k_l,2^{d-l-1}(2k_l+1)-1\right]}\left(A\right)}$} (anc,j) ~ add;
					cnot uaj | j, add;
					box {$U_{P_{\left[2^{d-l-1}(2k_l+1),2^{d-l}(k_l+1)-1\right]}\left(A\right)}$} (anc,j) | add;
					box {$U_{A^{\circ 2^{d-l-1}}}$} (uaanc,uaj) | add;
					cnot uaj | j, add;
					box {$R_Y\left(-\gamma_{l,k_l}\right)$} add;
					
					output {$\ket{i}$} j;
				\end{yquant}
			\end{tikzpicture}
		}
		\subcaption{Recursive quantum circuit framework for the binary-tree-based QHMF, where $U_{P_{\left[2^{d-l}k_l,2^{d-l-1}(2k_l+1)-1\right]}\left(A\right)}$ and $U_{P_{\left[2^{d-l-1}(2k_l+1),2^{d-l}(k_l+1)-1\right]}\left(A\right)}$ have the same circuit structure as $U_{P_{\left[2^{d-l}k_l,2^{d-l}(k_l+1)-1\right]}\left(A\right)}$, $k_l\in\{0,1,\cdots,2^l-1\}$, $l\in\{0,1,\cdots,d-2\}$, $b_l=(d-l-1)(n+1)-n+\sum_{k=0}^{d-2-l}a_k$.}
		\label{circuit: recursive quantum circuit}
	\end{minipage}
	
	\vspace{0.5cm}
	
	\begin{minipage}[b]{\textwidth}
		\centering
		\resizebox{\textwidth}{!}{
			\begin{tikzpicture}
				\begin{yquant}
					qubit {} add;
					qubit {} anc;
					qubit {$\ket{j}$} j;
					
					["north:$a_0$" 
					{font=\protect\footnotesize, inner sep=0pt}]
					slash anc;
					["north:$n$" 
					{font=\protect\footnotesize, inner sep=0pt}]
					slash j;
					
					hspace {5pt} -;
					box {$U_{P_{\left[2k_{d-1},2k_{d-1}+1\right]}\left(A\right)}$} (add,anc,j);
					hspace {5pt} -;
					
					text {$=$} (-);
					
					hspace {5pt} -;
					box {$R_Z\left(\phi_{k_{d-1}}\right)$} add;
					hspace {5pt} -;
					box {$R_Y\left(\theta_{k_{d-1}}\right)$} add;
					hspace {5pt} -;
					box {$U_{J}$} (anc,j) ~ add;
					hspace {5pt} -;
					box {$U_A$} (anc,j) | add;
					hspace {5pt} -;
					box {$R_Z\left(\varphi_{k_{d-1}}\right)$} add;
					hspace {5pt} -;
					box {$R_Y\left(-\theta_{k_{d-1}}\right)$} add;
					hspace {5pt} -;
					
					output {$\ket{i}$} j;
				\end{yquant}
			\end{tikzpicture}
		}
		\subcaption{Quantum circuit for the $1$-degree QHMF $P_{\left[2k_{d-1},2k_{d-1}+1\right]}\left(A\right)$ of the leaf layer, where $k_{d-1}\in\{0,1,\cdots,2^{d-1}-1\}$.}
		\label{circuit: 1-degree QHMF}
	\end{minipage}
	\caption{Quantum circuit for the binary-tree-based QHMF.}
	\label{circuit: binary-tree-based QHMF}
\end{figure*}
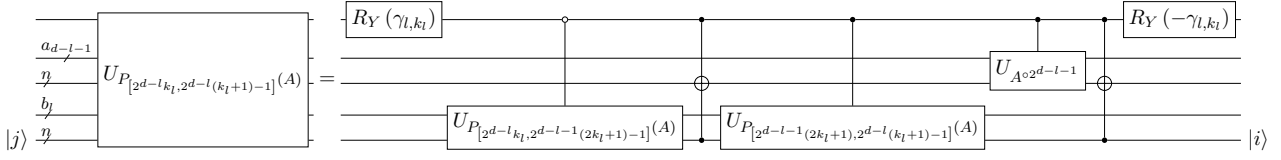
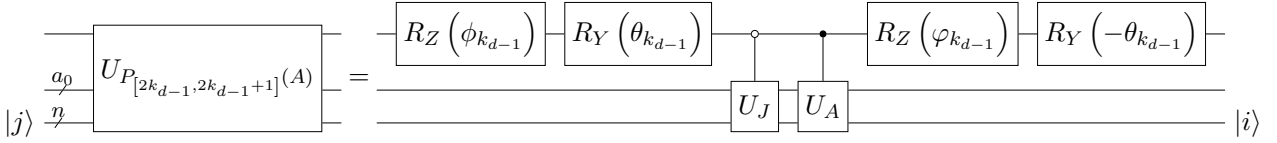

The following theorem states the main result for this subroutine.
\begin{theorem}[Binary-Tree-based QHMF]\label{theorem: binary-tree-based QHMF}
    Let $P(x)=\sum_{k=0}^{2^d-1}c_kx^k$ be a $\left(2^d-1\right)$-degree polynomial. In the power-oracle model, assume access to $(\alpha_l,a_l,\epsilon_l)$-block-encodings of $A^{\circ 2^l}\in\mathbb{C}^{2^n\times2^n}$ for $l\in\{0,1,\cdots,d-1\}$ together with their controlled versions. Then there exist four angle sequences 
	\begin{equation*}
		\begin{gathered}
			\left\{\gamma_{l,k_l} = 2\arccos\sqrt{\frac{\alpha_{l+1, 2k_l}}{\alpha_{l, k_l}}} \right\}_{k_l,l=0}^{2^{l}-1,d-2}, \\
			\left\{\theta_{k_{d-1}} = 2\arccos\sqrt{\frac{2^n\left|c_{2k_{d-1}}\right|}{\alpha_{d-1,k_{d-1}} }}\right\}_{k_{d-1}=0}^{2^{d-1}-1},\\
			\left\{\phi_{k_{d-1}} = -\operatorname{Arg}\left(c_{2k_{d-1}}\right) - \operatorname{Arg}\left(c_{2k_{d-1}+1}\right)\right\}_{k_{d-1}=0}^{2^{d-1}-1}, \\
			\left\{\varphi_{k_{d-1}} = -\operatorname{Arg}\left(c_{2k_{d-1}}\right) + \operatorname{Arg}\left(c_{2k_{d-1}+1}\right)\right\}_{k_{d-1}=0}^{2^{d-1}-1}, 
		\end{gathered}
	\end{equation*}
	such that the circuit shown in Fig.~\ref{circuit: binary-tree-based QHMF} implements the Hadamard matrix function $P\left(A\right)$ with a normalization factor 
	\begin{equation*}
        \alpha_{\rm BiT} = \sum_{k=0}^{2^d-1} 2^{n\overline{\operatorname{bit}_0(k)}}\alpha_0^{\operatorname{bit}_0(k)} |c_k|\prod_{j=1}^{d-1}\alpha_j^{\operatorname{bit}_j(k)}.
	\end{equation*}
\end{theorem}

The proof follows by induction on the recursion depth, using the LCU for combining two block-encodings and the QHP for multiplying with $A^{\circ 2^{l-1}}$. The details are given in Appendix~\ref{sec: proof for binary-tree-based QHMF}.

The circuit of binary-tree construction of binary-tree construction contains:
\begin{itemize}
    \item At tree level $l$ (with the root at $l=0$), there are $2^{d-l-1}$ nodes, each combining two children and requiring two controlled rotations ($4$ for $l=d-1$) and $2n$ multi-controlled $X$ gates (to implement the Hadamard product with $A^{\circ 2^{d-l-1}}$).
    \item The controlled rotations at level $l$ are controlled by $l$ qubits (the selection register) and can be implemented with depth $\mathcal{O}\left(\left(\log l\right)^3\right)$ and size $\mathcal{O}\left(l \left(\log l\right)^4\right)$ using the technique in~\cite{claudon2024poly}.
    \item The multi-controlled CNOT gates at level $l$ have $l+2$ control qubits, with similar depth $\mathcal{O}\left(\left(\log l\right)^3\right)$ and size $\mathcal{O}\left(l \left(\log l\right)^4\right)$.
    \item Queries to the block-encodings $U_{A^{\circ 2^{d-l-1}}}$ occur at level $l$; each encoding is called once per node with $l+1$ control qubits. The total number of queries to controlled $U_{A^{\circ 2^l}}$ for each $l$ is $2^{d-l-1}$.
\end{itemize}

Summing over all levels, we obtain the following complexity bounds; proof in Appendix~\ref{sec: proof of complexity for binary-tree QHMF}.
\begin{corollary}[Complexity of Binary-Tree-based QHMF]\label{cor: complexity of binary-tree-based QHMF}
	The binary-tree-based QHMF in Theorem~\ref{theorem: binary-tree-based QHMF} can be implemented with:
    \begin{itemize}
        \item Query complexity:  
        \begin{equation*}
            Q_{d-l}^l = \mathcal{O}\left(2^{d-l}\right)
        \end{equation*}
        for $l\in\{0,1,\cdots,d-1\}$.
        \item Additional gates: 
        \begin{align*}
            D = \mathcal{O}\left(n2^d\left(\log d\right)^3\right), \quad 
            S = \mathcal{O}\left(n2^dd\left(\log d\right)^4\right).
        \end{align*}
        \item Ancilla count: 
        \begin{equation*}
            N=\mathcal{O}\left(nd + \sum_{l=0}^{d-1}a_l\right).
        \end{equation*}
    \end{itemize}
\end{corollary}

These bounds show that the binary-tree subroutine is exponentially better in ancilla than the full factorization, but its depth is exponential in $d$. Since the degree is small, the depth overhead is acceptable.
In the trade-off framework, we can choose $m$ large enough to keep $K$ moderate, so that the binary-tree subroutine becomes efficient.

\subsection{Resource Analysis and Comparison}\label{subsec: comparison}
We now consolidate the results of the preceding subsections into a unified resource characterization of the entire trade-off framework. Recall that the target is a Hadamard matrix function $P(A)$ of degree $2^d-1$. The user chooses an integer $m$ with $1 \le m \le 2^d-1$, and then decomposes $P$ as $P(x) = \prod_{s=0}^{m-1} P^{(s)}(x)$, where each $\deg P^{(s)} \leq K = \left\lceil \frac{2^d-1}{m} \right\rceil$, and implements each $P^{(s)}(A)$ using the binary-tree subroutine. The $m$ factors are then combined via the QHP. The overall circuit architecture is illustrated in Fig.~\ref{circuit: trade-off}.

The following theorem summarizes the resource costs of the complete construction as explicit functions of $m$, $n$, and $d$.
\begin{theorem}[Complexity of Trade-Off QHMF]\label{thm: complexity of trade-off QHMF}
    Let $m$ be an integer with $1\leq m \leq 2^d-1$. Set $K = \left\lceil \frac{2^d-1}{m} \right\rceil$. Let $P(x)=\prod_{s=0}^{m-1}P^{(s)}(x)$ be a $\left(2^d-1\right)$-degree polynomial, where $P^{(s)}(x)=\sum_{k=0}^{K}c^{(s)}_kx^k$. Then the trade-off QHMF in Theorem~\ref{thm: trade-off QHMF} can be implemented with:
    \begin{itemize}
        \item Query complexity: \\
        \begin{equation*}
            Q_{\left\lceil\log (K+1)\right\rceil-l}^l = \mathcal{O}\left(m2^{d-\log m-l}\right)
        \end{equation*} 
        for $l\in\left\{0,1,\cdots,\left\lceil\log K\right\rceil-1\right\}$.
        \item Additional gates:\\
        \begin{align*}
            D(m) &= \mathcal{O}\left(\frac{n2^d}{m} \left(\log\log\frac{2^d}{m}\right)^3 + \log m \right), \\
            S(m) &= \mathcal{O}\left(n2^d \log\frac{2^d}{m}\left(\log\log\frac{2^d}{m}\right)^4 + nm \right).
        \end{align*}
        \item Ancilla count:\\
        \begin{equation*}
            N(m) = \mathcal{O}\left(m\left(n \log\frac{2^d}{m} + \sum\limits_{k=0}^{\left\lceil\log K\right\rceil - 1} a_k\right)\right).
        \end{equation*}
        \item Normalization factor:\\
        \begin{equation*}
            \begin{aligned}
                & \alpha_{\rm TO}(m) \\
                =& \prod\limits_{s=0}^{m-1} \left(\sum\limits_{k=0}^{K} 2^{n\overline{\operatorname{bit}_0(k)}}\alpha_0^{\operatorname{bit}_0(k)} \left|c^{(s)}_k\right|\prod\limits_{j=1}^{\left\lceil\log K\right\rceil-1}\alpha_{j}^{\operatorname{bit}_j(k)}\right).
            \end{aligned}
        \end{equation*}
    \end{itemize}
\end{theorem}
\begin{proof}
    The result is directly obtained from Theorem~\ref{thm: trade-off QHMF} and Corollary~\ref{cor: complexity of binary-tree-based QHMF}.
\end{proof}

\begin{table*}[htbp]
	\centering
	\caption{Resource comparison for implementing a $(2^d-1)$-degree Hadamard matrix function $P(A)=\sum_{k=0}^{2^d-1}c_kA^{\circ k}$. Here, $U_{A^{\circ 2^l}}$ denotes an $(\alpha_l,a_l,\epsilon_l)$-block-encoding of $A^{\circ 2^l}$ for $l\in\{0,1,\cdots,d-1\}$, and $\left\{r_k\right\}_{k=0}^{2^d-2}$ are roots of $P(x)$. The trade-off parameter $m$ determines the factor degree $K=\left\lceil\frac{2^d-1}{m}\right\rceil$, so that $P(A)=\bigcirc_{s=0}^{m-1}P^{(s)}(A)$ with $P^{(s)}(A)=\sum_{k=0}^{K}c^{(s)}_kA^{\circ k}$. 
    }  
	\label{tab: comparison of QHMF}
    \resizebox{\linewidth}{!}{
    \begin{tabular}{|c|c|c|c|c|c|}
        \hline
        \multicolumn{2}{|c|}{\textbf{QHMF}} & \makecell{\textbf{Factorization} \\ (Theorem~\ref{thm: factorization-based QHMF}, Corollary~\ref{cor: complexity of factorization-based QHMF})} & \makecell{\textbf{Trade-off} \\ (Theorem~\ref{thm: trade-off QHMF} and~\ref{thm: complexity of trade-off QHMF}) } & \makecell{\textbf{Binary tree} \\ (Theorem~\ref{theorem: binary-tree-based QHMF}, Corollary~\ref{cor: complexity of binary-tree-based QHMF})} & \makecell{\textbf{LCU} \\ (\cite{guo2025quantum}, Lemma~\ref{lem: complexity of LCU-based QHMF})} \\
        \hline
        \multicolumn{2}{|c|}{\makecell{\textbf{Query} \\ \textbf{Model}}} & \makecell{single-oracle \\/ power-oracle} & \multicolumn{3}{c|}{power-oracle} \\
        \hline
        \multicolumn{2}{|c|}{\makecell{\textbf{Query to} \\ \textbf{Oracle}}} & \makecell{$Q_1^0=\Omega\left(2^d\right)$ \\ (in parallel)} & \makecell{$Q_{\left\lceil\log (K+1)\right\rceil-l}^l = \mathcal{O}\left(m2^{d-\log m-l}\right)$ \\ for $l\in\{0,1,\cdots,\left\lceil\log K\right\rceil-1\}$} & \makecell{$Q_{d-l}^l = \mathcal{O}\left(2^{d-l}\right)$ \\ for $l\in\left\{0,1,\cdots,d-1\right\}$} & \makecell{$Q_d^l = \mathcal{O}\left(2^d\right)$ \\ for $l\in\left\{0,1,\cdots,d-1\right\}$} \\
        \hline
        \multirow{2}{*}{\makecell{\textbf{Additional} \\ \textbf{Gate}}} & \textbf{Size} & \makecell{$\mathcal{O}\left(n2^d\right)$} & \makecell{$\mathcal{O}\left(n2^d \log\frac{2^d}{m}\left(\log\log\frac{2^d}{m}\right)^4 + nm \right)$} & \makecell{$\mathcal{O}\left(n2^dd\left(\log d\right)^4\right)$} & \makecell{$\mathcal{O}\left(n2^dd^2\left(\log d\right)^4\right)$} \\
        \cline{2-6}
        & \textbf{Depth} & \makecell{$\mathcal{O}\left(n+d\right)$} & $\mathcal{O}\left(\frac{n2^d}{m} \left(\log\log\frac{2^d}{m}\right)^3 + \log m \right)$ & $\mathcal{O}\left(n2^d\left(\log d\right)^3\right)$ & \makecell{$\mathcal{O}\left(n2^dd\left(\log d\right)^3\right)$} \\
        \hline
        \multicolumn{2}{|c|}{\makecell{\textbf{Ancilla} \\ \textbf{Count}}} & \makecell{$\mathcal{O}\left(2^d(n+a_0)\right)$} & $\mathcal{O}\left(m\left(n \log\frac{2^d}{m} + \sum\limits_{k=0}^{\left\lceil\log K\right\rceil - 1} a_k\right)\right)$ & \makecell{$\mathcal{O}\left(nd + \sum_{l=0}^{d-1}a_l\right)$} & \makecell{$\mathcal{O}\left(nd + \sum_{l=0}^{d-1}a_l\right)$} \\
        \hline
        \multicolumn{2}{|c|}{\makecell{\textbf{Normalization} \\ \textbf{Factor}}} & $\alpha_{\rm Fac}$ & $\alpha_{\rm TO}(m)$ & $\alpha_{\rm BiT}$ & $\alpha_{\rm LCU}$\\
        \hline
    \end{tabular}
    }
    \begin{align*}
        \alpha_{\rm Fac} =& \left|c_{2^d-1}\right|\prod\limits_{k=0}^{2^d-2}\left(\alpha_0+2^n\left|r_k\right|\right)
        & \alpha_{\rm TO}(m) =& \prod\limits_{s=0}^{m-1} \left(\sum\limits_{k=0}^{K} 2^{n\overline{\operatorname{bit}_0(k)}}\alpha_0^{\operatorname{bit}_0(k)} \left|c^{(s)}_k\right|\prod\limits_{j=1}^{\left\lceil\log K\right\rceil-1}\alpha_{j}^{\operatorname{bit}_{j}(k)}\right) \\
        \alpha_{\rm BiT} =& \sum_{k=0}^{2^d-1} 2^{n\overline{\operatorname{bit}_0(k)}}\alpha_0^{\operatorname{bit}_0(k)} |c_k|\prod_{j=1}^{d-1}\alpha_j^{\operatorname{bit}_j(k)}
        & \alpha_{\rm LCU} =& 2^n\left|c_0\right| + \sum\limits_{k=1}^{2^d-1}\left|c_k\right|\prod\limits_{j=0}^{d-1} \alpha_j^{\operatorname{bit}_j(k)}
    \end{align*}
\end{table*}

The formulas in Theorem~\ref{thm: complexity of trade-off QHMF} reveal the explicit depth–ancilla trade-off. Increasing $m$ decreases the additional depth while increasing the query complexity, additional size, and ancilla count.

Before comparing the resource counts, we emphasize that Table~\ref{tab: comparison of QHMF} reports resources within the oracle model indicated in its ``Query Model'' row. 
In the single-oracle model, only an $(\alpha_0,a_0,\epsilon_0)$-block-encoding $U_A$ of $A$ and its controlled version are available. 
Higher Hadamard powers $A^{\circ 2^l}$ for $l\ge1$ are not primitive oracles; if a construction uses them, they must be synthesized from $U_A$, and this cost must be included. 
The factorization construction queries only the single-qubit controlled $U_A$. 
In the power-oracle model, the block encodings $U_{A^{\circ 2^l}}$ are treated as primitive oracles and their implementation cost is not included in the table. 
The trade-off, binary-tree, and LCU constructions query $(l+1)$- or $d$-qubit controlled versions of these oracles. 
Consequently, the query complexity is only comparable within the same oracle model. 
Resource counts across different oracle models, including depth, size, ancilla, and normalization factor, are not directly comparable unless the cost of synthesizing the power oracles is explicitly included. 
The all-ones matrix $J$ is not treated as an oracle; its exact block encoding is included in the non-oracle gate and ancilla counts.

Table~\ref{tab: comparison of QHMF} summarizes the resource costs within each oracle model. 
In the single-oracle model, the factorization method achieves the logarithmic circuit depth $\mathcal{O}(n+d)$ and saturates the query lower bound, but at the cost of an exponential number of ancilla and a large normalization factor.
In the power-oracle model, the trade-off framework decomposes the polynomial into $m$ lower-degree factors and combines them via the Hadamard product. 
This yields a discrete depth-ancilla trade-off: $m=2^d-1$ recovers the factorization-based QHMF, while $m=1$ reduces to the binary-tree method applied to the full polynomial. 
The binary-tree method serves as the subroutine for low-degree factors.
Compared with the LCU method in the same power-oracle model, it improves the depth, size, and the query complexity to higher-order powers while keeping the ancilla count comparable. 
Because the factorization method operates in the single-oracle model, its query complexity should not be compared directly with those of the power-oracle constructions; such a comparison requires including the cost of synthesizing $U_{A^{\circ 2^l}}$ from $U_A$.

\section{Error Analysis}\label{sec: error}
In Sec.~\ref{subsec: resource metrics}, we defined the approximation error $\epsilon_{\rm appro}$, the absolute and relative implementation errors $\epsilon_{\rm abs}$ and $\epsilon_{\rm rel}$, and the post-selection success probability $p_{\rm succ}$. 
Here we analyze how these quantities propagate through the three QHMF constructions. 
We first discuss the polynomial approximation error in Sec.~\ref{subsec: approximation error}, and then derive explicit implementation-error bounds for the factorization, binary-tree, and trade-off constructions in Sec.~\ref{subsec: implementation error}.

\subsection{Approximation Error of Target Function}\label{subsec: approximation error}
As defined in Sec.~\ref{subsec: resource metrics}, the approximation error is
\begin{equation*}
    \epsilon_{\rm appro}=\max_{z\in\Omega}|f(z)-P(z)|.
\end{equation*}
Because $P$ acts entry-wise on $A$, this scalar error propagates to each matrix element without amplification, so that for any matrix $A$ with all entries in $\Omega$,
\begin{equation*}
    \|f(A)-P(A)\|_{\max}\le \epsilon_{\rm appro}.
\end{equation*}

The choice of $P$ determines both the approximation error and the required polynomial degree, which in turn fixes the parameters $d$ (and hence the circuit resources). We consider three standard strategies.

\subsubsection{Chebyshev Approximation}
For a continuous $f$ on a compact real interval $[a,b]$, the Chebyshev expansion of degree $K$,
\begin{equation*}
    P^{\text{Cheb}}_K(x) = \sum_{k=0}^{K} c_k T_k\left(\frac{2x-a-b}{b-a}\right),
\end{equation*}
provides a near-optimal uniform approximation, where $T_k$ is the Chebyshev polynomial of the first kind. When $f$ is analytic on a neighborhood of $[a,b]$, the Chebyshev coefficients decay geometrically, yielding
\begin{equation}\label{eq: cheb-bound}
    E_\infty\left(P^{\text{Cheb}}_K, f, [a,b]\right) \leq C \rho^{-K},
\end{equation}
for some $\rho>1$ and a constant $C = C\left(f[a,b], \rho\right)$~\cite{trefethen2013approximation, rivlin1981, mason2002chebyshev}. If $f$ has singularities on or near $[a,b]$, the decay degrades to algebraic, which explains the slow convergence observed for functions~\cite{wang2018convergence}, such as the function $x^{0.4}$ at the origin (Sec.~\ref{sec: image processing}).

\subsubsection{Minimax Approximation}
The optimal polynomial in the uniform norm is the minimax polynomial $P^*_K$, defined by
\begin{equation*}
    E_\infty\left(P^*_K, f, \Omega\right) = \min_{\deg Q\le K} \max_{x\in\Omega} \left| f(x) - Q(x) \right|.
\end{equation*}
For real intervals, $P^*_K$ is uniquely characterized by the equioscillation of the error at $K+2$ extremal points (Chebyshev alternation theorem) and can be computed by the Remez algorithm~\cite{remez1934determination, remez1934procede, veidinger1960numerical, powell1981approximation}. In the complex plane, an analogous complex Remez algorithm applies when $\Omega$ is sufficiently regular and $f$ is analytic; the error then attains its maximum modulus on the boundary of $\Omega$ and satisfies a complex equioscillation condition~\cite{hubner2025computing, tang1988fast, lebailly1999computing}.

\subsubsection{Choice of Degree}
Given a target accuracy $\varepsilon$, one selects the minimal degree $K$ such that $E_\infty(P_K, f, \Omega)\le\varepsilon$. The degree $K$ then determines the exponent $d = \left\lceil \log(K+1) \right\rceil$ in the Hadamard-polynomial framework. This directly couples the approximation accuracy to the circuit resources of Sec.~\ref{sec: QHMF}: larger $K$ requires larger $d$, hence more ancilla, depth, and queries. In particular, when $f$ is analytic and Eq.~\eqref{eq: cheb-bound} holds, the degree grows only logarithmically in $1/\epsilon_{\rm appro}$, so that
\begin{equation*}
    d = \mathcal{O}\left(\log \log \frac{1}{\epsilon_{\rm appro}}\right),
\end{equation*}
a mild dependence that makes the framework attractive even for high accuracy targets.

\subsection{Implementation Error of Quantum Circuit}\label{subsec: implementation error}
Recall from Sec.~\ref{subsec: resource metrics} that a QHMF circuit block-encodes $P(A)$ with normalization factor $\alpha$, absolute implementation error $\epsilon_{\rm abs}$, relative implementation error $\epsilon_{\rm rel}$, and post-selection success probability $p_{\rm succ}$. 
The relative error determines the observable accuracy after post-selection, while $p_{\rm succ}$ determines the sampling cost. 
We now derive explicit bounds on these quantities for the three constructions. 

Each Hadamard power $A^{\circ 2^l}$ is accessed through an $\left(\alpha_l,a_l,\epsilon_l\right)$-block-encoding, where $\epsilon_l$ bounds the operator-norm error of the encoding. By Corollary~\ref{corollary: QHP of m matrices}, if $\widetilde A_k$ approximates $A_k$ with error $\epsilon_k$ and normalization $\alpha_k$, then the error of QHP is given by
\begin{equation}\label{eq: QHP error}
    \left\| \bigcirc_{k=0}^{m-1}A_k - \bigcirc_{k=0}^{m-1}\widetilde A_k \right\| \leq \left( \sum_{k=0}^{m-1}\frac{\epsilon_k}{\alpha_k} \right) \prod_{k=0}^{m-1}\alpha_k.
\end{equation}
For a linear combination, the standard LCU error propagation gives
\begin{equation}\label{eq: lcu error}
    \left\|
    \sum_{k=0}^{m-1}c_kA_k
    -
    \sum_{k=0}^{m-1}c_k\widetilde A_k
    \right\|
    \le
    \sum_{k=0}^{m-1}|c_k|\epsilon_k,
\end{equation}
provided the state-preparation unitaries are exact~\cite{gilyen2019quantum}. 

In the following, we analyze $\epsilon_{\rm abs}$, $\epsilon_{\rm rel}$, and $p_{\rm succ}$ for the factorization, binary-tree, and trade-off constructions.

\subsubsection{Error in Factorization Construction}
For the factorization-based QHMF, each linear factor $A-r_kJ$ is
block-encoded by an LCU using $U_A$ and the exact $U_J$.
Its normalization is $\alpha_0+2^n|r_k|$, and its block-encoding error is $\epsilon_0$. 

Applying the QHP error bound of Corollary~\ref{corollary: QHP of m matrices} to the $2^d-1$ factors and multiplying by the scalar $c_{2^d-1}$, the absolute implementation error is
\begin{equation*}\label{eq: error QHMF factorization}
    \epsilon_{\rm abs}^{\rm Fac} = \left|c_{2^d-1}\right| \epsilon_0
    \sum_{k=0}^{2^d-2}\prod_{l\ne k}\left(\alpha_0+2^n\left|r_l\right|\right).
\end{equation*}
The corresponding relative implementation error is 
\begin{equation*}\label{eq: relative error QHMF factorization}
    \epsilon_{\rm rel}^{\rm Fac} 
    = \frac{\epsilon_{\rm abs}^{\rm Fac}}{\alpha_{\rm Fac}}
    = \epsilon_0 \sum_{k=0}^{2^d-2} \frac{1}{\alpha_0+2^n|r_k|}.
\end{equation*}
The post-selection success probability is
\begin{equation*}
    \begin{aligned}
        p_{\rm succ}^{\rm Fac} 
        = \frac{\|P(A)\|_F^2}{2^n|c_{2^d-1}|^2 \prod_{k=0}^{2^d-2}(\alpha_0+2^n|r_k|)^2}.
    \end{aligned}
\end{equation*}

The magnitude of this error is controlled by the normalization factors. If $\alpha_0=\mathcal{O}(2^n)$ and $|r_k|=\mathcal{O}(1)$, then $\alpha_0+2^n|r_k| = \mathcal{O}(2^n)$, and $\alpha_{\rm Fac} = \mathcal{O}\left(2^{n(2^d-1)}\right)$. The absolute implementation error scales as
\begin{equation*}
    \epsilon_{\rm abs}^{\rm Fac} = \mathcal{O}\left(2^{n(2^d-2)+d}\epsilon_0\right),
\end{equation*}
which grows exponentially in $n$ and double-exponentially in $d$. 
For the relative error, the same assumption gives
\begin{equation*}
    \epsilon_{\rm rel}^{\rm Fac} = \mathcal{O}\left(\frac{2^d\epsilon_0}{2^n}\right).
\end{equation*}
Thus, the absolute implementation error is large, but the relative error remains small whenever $2^d\epsilon_0 \ll 2^n$. 
This is a consequence of the large normalization factor. 
The post-selection success probability scales as
\begin{equation*}
    p_{\rm succ}^{\rm Fac} = \mathcal{O}\left(\frac{\|P(A)\|_F^2}{2^{n(2^{d+1}-1)}}\right)
\end{equation*}
The dominant practical limitation of the factorization construction is, therefore, not the relative accuracy, but this exponentially small post-selection success probability. 
This illustrates why the absolute error alone is insufficient to characterize the performance of the construction.

\subsubsection{Error in Binary Tree Construction}
The implementation error of the binary-tree QHMF can be obtained from the recursive circuit in Fig.~\ref{circuit: binary-tree-based QHMF}. 
For the leaf layer $l=d-1$, the polynomial is
\begin{equation*}
    P_{[2k_{d-1},2k_{d-1}+1]}(A)=c_{2k_{d-1}}J+c_{2k_{d-1}+1}A,
\end{equation*}
with normalization factor $\alpha_{d-1,k_{d-1}}$. The all-ones matrix $J$ is block-encoded exactly, while $U_A$ carries an error $\epsilon_0$.
Then the error of the leaf block-encoding is
\begin{equation}\label{eq: leaf error}
    \epsilon_{d-1,k_{d-1}} = \left|c_{2k_{d-1}+1}\right| \epsilon_0.
\end{equation}

For a parent node $l\in\{0,1,\cdots,d-2\}$ with index $k_l\in\{0,1,\cdots,2^l-1\}$, the two children are $2k_l$ and $2k_l+1$. The parent polynomial is obtained by the LCU 
\begin{equation*}
    \begin{aligned}
        & P_{[2^{d-l}k_l,2^{d-l}(k_l+1)-1]}(A) \\
        =& P_{[2^{d-l}k_l,2^{d-l-1}(2k_l+1)-1]}(A) \\
        & + P_{[2^{d-l-1}(2k_l+1),2^{d-l}(k_l+1)-1]}(A)
        \circ A^{\circ 2^{d-l-1}} .
    \end{aligned}
\end{equation*}
where the normalization factors obey Eq.~\eqref{eq: interval normalization factor}. Using the LCU error propagation and the QHP error bound, the error for implementing $P_{[2^{d-l}k_l, 2^{d-l}(k_l+1)-1]}(A)$ satisfies
\begin{equation}\label{eq: BT error recursion}
    \epsilon_{l,k_l} = \epsilon_{l+1,2k_l} + \epsilon_{l+1,2k_l+1}\alpha_{d-l-1} + \epsilon_{d-l-1} \alpha_{l+1,2k_l+1}.
\end{equation}

Equations~\eqref{eq: leaf error} and~\eqref{eq: BT error recursion} form a linear recursion. To solve it, define the propagation coefficient $C_{l,k_l}$ from a node $(l,k_l)$ to the root $(0,0)$. By construction, we have $C_{0,0}=1$, and for $l=0,1,\dots,d-2$, we have $C_{l+1,2k_l}=C_{l,k_l}$ and $C_{l+1,2k_l+1}=\alpha_{d-l-1}C_{l,k_l}$.
Equivalently, $C_{l,k_l}=\prod_{j=0}^{l-1}\alpha_{d-j-1}^{\operatorname{bit}_j(k_l)}$. The absolute implementation error of block encoding the root polynomial is, therefore,
\begin{equation}\label{eq: error QHMF binary tree}
    \begin{aligned}
        \epsilon_{\rm abs}^{\rm Bit} = \epsilon_{0,0}
        =& \epsilon_0 \sum_{k=0}^{2^{d-1}-1} \left|c_{2k+1}\right| \prod_{j=0}^{d-2}\alpha_{d-j-1}^{\operatorname{bit}_j(k)}
        \\
        & + \sum_{l=0}^{d-2} \epsilon_{d-l-1} \sum_{k_l=0}^{2^l-1} \alpha_{l+1,2k_l+1} \prod_{j=0}^{l-1}\alpha_{d-j-1}^{\operatorname{bit}_j(k_l)}.
    \end{aligned}
\end{equation}
This expression is a positive linear combination of the block-encoding errors $\epsilon_l$ ($l\in\left\{0,1,\dots,d-1\right\}$). The weight of each $\epsilon_l$ is determined by the interval normalization factors and the power-oracle normalization factors.

The corresponding relative implementation error is obtained by dividing by the normalization factor $\alpha_{\rm BiT}$, that is, $\epsilon_{\rm rel}^{\rm Bit} = \frac{\epsilon_{\rm abs}^{\rm Bit}}{\alpha_{\rm BiT}}$.
The post-selection success probability is $p_{\rm succ}^{\rm Bit} = \frac{\|P(A)\|_F^2}{\alpha_{\rm BiT}^22^n}$.

The scaling of $\epsilon_{\rm abs}^{\rm Bit}$ is controlled jointly by the number of data qubits $n$ and the degree parameter $d$. Inspecting Eq.~\eqref{eq: error QHMF binary tree}, each product over $j$ contains at most $d-1$ normalization factors. Under the assumption $\alpha_l=\mathcal{O}(2^n)$ for all $l$, each such factor is $\mathcal{O}(2^n)$, so each product is $\mathcal{O}(2^{n(d-1)})$. The first sum has $2^{d-1}$ terms, giving a contribution $\mathcal{O}\left(2^{(n+1)(d-1)}\epsilon_0\right)$.
For the second sum, there are $d-1$ outer terms. In the $l$-th outer term, the inner sum has $2^l$ terms. Each inner term contains $l+1$ normalization factors: the factor $\alpha_{l+1,2k_l+1}$ together with the $l$ factors in the product $\prod_{j=0}^{l-1}\alpha_{d-j-1}^{\operatorname{bit}_j(k_l)}$. Therefore, each inner term is bounded by $\mathcal{O}\left(2^{n(l+1)}\epsilon_0\right)$, where we assume $\epsilon_l=\mathcal{O}\left(\epsilon_0\right)$ for all $l$. Summing over the $2^l$ inner terms gives $\mathcal{O}\left(2^{n+l(n+1)}\epsilon_0\right)$.
The outer sum over $l\in\left\{0,\dots,d-2\right\}$ is dominated by the term $l=d-2$, yielding $\mathcal{O}\left(2^{d(n+1)-2}\epsilon_0\right)$.
Absorbing the constant factor into the asymptotic notation, the total implementation error scales as
\begin{equation}\label{eq: BT error scaling}
    \epsilon_{\rm abs}^{\rm Bit} = \mathcal{O}\left(2^{(n+1)d}\epsilon_0\right),
\end{equation}
which exhibits exponential growth in both $n$ and $d$. 

For the relative error, the same assumption $\alpha_l=\mathcal{O}(2^n)$ gives $\alpha_{\rm BiT}=\mathcal{O}\left(2^{n(d+1)}\right)$. 
Combining this with Eq.~\eqref{eq: BT error scaling}, we obtain
\begin{equation*}
    \epsilon_{\rm rel}^{\rm Bit} = \mathcal{O}\left(\frac{2^d\epsilon_0}{2^n}\right).
\end{equation*}
Thus, the relative error grows exponentially in $d$ but decreases exponentially in $n$. 
As in the factorization construction, the absolute error is large but the relative error is much more favorable. 
The post-selection success probability scales as
\begin{equation*}
    p_{\rm succ}^{\rm Bit} = \Theta\left(\frac{\|P(A)\|_F^2}{2^{n(2d+3)}}\right),
\end{equation*}
which decreases exponentially with $n$ and $d$, but is significantly larger than the factorization success probability $p_{\rm succ}^{\rm Fac}$. 
This quantifies the advantage of the binary-tree construction: it trades a larger circuit depth for a much better post-selection success probability, while keeping the relative error comparable to that of the factorization method.

\subsubsection{Error in Trade-Off Construction}
Let $P(x)=\prod_{s=0}^{m-1}P^{(s)}(x)$ with $\deg P^{(s)}\le K=\left\lceil\frac{2^d-1}{m}\right\rceil$, and let the binary-tree subroutine implement each $P^{(s)}(A)$ as an
$(\alpha^{(s)},a^{(s)},\epsilon^{(s)})$-block-encoding. The normalization factor $\alpha^{(s)}$ is obtained from the binary-tree normalization formula applied to the factor $P^{(s)}$, namely,
\begin{equation*}
    \alpha^{(s)} = \sum_{k=0}^{K} 2^{n\overline{\operatorname{bit}_0(k)}} \alpha_0^{\operatorname{bit}_0(k)} \left|c^{(s)}_k\right| \prod_{j=1}^{\left\lceil\log K\right\rceil-1} \alpha_j^{\operatorname{bit}_j(k)},
\end{equation*}
while the implementation error $\epsilon^{(s)}$ is given by Eq.~\eqref{eq: error QHMF binary tree}. Explicitly, with $d_s=\left\lceil\log(K+1)\right\rceil$,
\begin{equation*}
    \begin{aligned}
        \epsilon^{(s)} =& \epsilon_0 \sum_{k=0}^{2^{d_s-1}-1} \left|c^{(s)}_{2k+1}\right| \prod_{j=0}^{d_s-2} \alpha_{d_s-j-1}^{\operatorname{bit}_j(k)} \\
        & + \sum_{l=0}^{d_s-2} \epsilon_{d_s-l-1} \sum_{k_l=0}^{2^l-1} \alpha_{l+1,2k_l+1} \prod_{j=0}^{l-1} \alpha_{d_s-j-1}^{\operatorname{bit}_j(k_l)}.
    \end{aligned}
\end{equation*}

By the QHP error bound in Corollary~\ref{corollary: QHP of m matrices}, the total absolute implementation error is
\begin{equation}\label{eq: trade impl error}
    \epsilon_{\rm abs}^{\rm TO}(m) = \sum_{s=0}^{m-1} \epsilon^{(s)} \prod_{t\ne s}\alpha^{(t)} = \alpha_{\rm TO} \sum_{s=0}^{m-1} \frac{\epsilon^{(s)}}{\alpha^{(s)}} .
\end{equation}
The corresponding relative implementation error is 
\begin{equation*}
    \epsilon_{\rm rel}^{\rm TO}(m) 
    = \frac{\epsilon_{\rm abs}^{\rm TO}(m)}{\alpha_{\rm TO}(m)}
    = \sum_{s=0}^{m-1} \frac{\epsilon^{(s)}}{\alpha^{(s)}},
\end{equation*}
and the post-selection success probability is $p_{\rm succ}^{\rm TO}(m) = \frac{\|P(A)\|_F^2}{\alpha_{\rm TO}(m)^22^n}$.

We now analyze the scaling of these three quantities under the assumption $\alpha_l=\mathcal{O}(2^n)$ and $\epsilon_l=\mathcal{O}\left(\epsilon_0\right)$ for all $l$. 
Each factor $P^{(s)}$ has binary-tree depth $d_s=\Theta(d-\log m)$, and its normalization factor and implementation error satisfy $\alpha^{(s)} = \mathcal{O}\left(2^{n(d-\log m + 1)}\right)$ and $\epsilon^{(s)} = \mathcal{O}\left(2^{(n+1)(d-\log m)}\epsilon_0\right)$. Consequently, $\alpha_{\rm TO}=\prod_{s=0}^{m-1}\alpha^{(s)}=\mathcal{O}\left(2^{mn(d-\log m + 1)}\right)$.
Substituting these estimates into Eq.~\eqref{eq: trade impl error}, we find that the absolute implementation error scales as
\begin{equation*}
    \epsilon_{\rm abs}^{\rm TO}(m) = \mathcal{O}\left(2^{mn(d-\log m+1)+d-n}\epsilon_0\right),
\end{equation*}
while the relative implementation error scales as
\begin{equation*}
    \epsilon_{\rm rel}^{\rm TO}(m) = \mathcal{O}\left(\frac{2^{d}\epsilon_0}{2^{n}}\right).
\end{equation*}
The post-selection success probability scales as
\begin{equation}\label{eq: trade success scaling}
    p_{\rm succ}^{\rm TO}(m)
    = \Theta\left(\frac{\|P(A)\|_F^2}{2^{2m n (d-\log m+1)+n}}\right).
\end{equation}

The trade-off parameter $m$ thus controls the error through two competing effects. 
Increasing $m$ reduces the degree $K$ of each factor, and hence the circuit depth of each binary-tree subroutine, but it also increases the number of factors and the total normalization factor $\alpha_{\rm TO}(m)$. 
As a result, the relative error remains essentially constant, whereas the absolute error and the sampling cost grow exponentially with $m$. 
This shows that the dominant limitation on $m$ is not the relative accuracy but the post-selection success probability. 

For a target accuracy $\varepsilon$, the admissible values of $m$ must satisfy
\begin{equation*}
    \epsilon_{\rm appro} + \epsilon_{\rm abs}^{\rm TO}(m) \leq \varepsilon.
\end{equation*}
In addition, if the sampling cost is constrained, one requires $p_{\rm succ}^{\rm TO}(m) \ge p_{\min}$.
Using Eq.~\eqref{eq: trade success scaling}, the success-probability constraint is equivalent to the implicit inequality
\begin{equation*}
    m\left(d-\log m+1\right) \leq \frac{1}{2n}\log\left(\frac{\|P(A)\|_F^2}{p_{\min}}\right) - n.
\end{equation*}
Because the left-hand side contains both $m$ and $\log m$, this inequality has no elementary closed-form solution for $m$. 
We therefore define $m_{\max}$ as the largest integer $m\in\{1,2,\dots,2^d-1\}$ satisfying this inequality, which can be found by a simple numerical search or by a bisection on the monotonic branch of $m(d-\log m+1)$ for $m\leq 2^{d-1}$.
In summary, the trade-off framework interpolates between the two error regimes as $m$ varies over $\{1,\dots,2^d-1\}$, and the choice of $m$ must balance depth, width, and sampling cost simultaneously.

\section{Applications}\label{sec: applications}
To demonstrate the practical utility of our quantum Hadamard matrix function framework, we apply it to two representative tasks: approximating nonlinear activation functions for quantum neural networks, and performing pointwise image transformations. Both tasks require applying the same polynomial function to every entry of an input matrix.

\subsection{Activation Functions in Quantum Neural Networks}\label{subsec: activation function}
Nonlinear activation functions are indispensable in classical neural networks. In the quantum setting, the same requirement appears when one wishes to realize a quantum analog of a feed-forward layer: after a linear transformation, the amplitudes must be passed through a nonlinear map. Because any continuous activation can be approximated uniformly by a polynomial on a compact interval, the Hadamard-polynomial framework developed in Sec.~\ref{sec: QHMF} supplies a systematic route to this task.

We consider the two most common bounded activations, the sigmoid function
\begin{equation*}
    \sigma(x) = \frac{1}{1+e^{-x}},
\end{equation*}
and the hyperbolic tangent function 
\begin{equation*}
    \operatorname{tanh}(x) = \frac{e^x - e^{-x}}{e^x + e^{-x}}.
\end{equation*}
For both functions, the effective input interval is typically restricted to $[-4,4]$, outside which the functions saturate. We compute the maximum absolute error (MaxAE) and the root mean square error (RMSE) of polynomials of varying degrees with power base coefficients obtained by Chebyshev approximations in Table~\ref{tab: errors of sigmoid tanh}.
\begin{table}[htbp]
    \centering
    \caption{MaxAE and RMSE for Chebyshev approximations of the sigmoid and tanh functions on $[-4,4]$. Errors are evaluated on $20001$ uniformly spaced points. The bold entries are the lowest degrees satisfying $\epsilon_{\rm appro}\leq10^{-2}$.}
    \label{tab: errors of sigmoid tanh}
    \resizebox{\linewidth}{!}{
    \begin{tabular}{|c|c|c|c|c|}
        \hline
        \textbf{Function} & \textbf{Degree} & $\mathbf{d}$ & \textbf{MaxAE} & \textbf{RMSE} \\
        \hline
        \multirow{4}{*}{Sigmoid} & $1$ & $1$ & $1.462\times10^{-1}$ & $6.629\times10^{-2}$ \\
        & $3$ & $2$ & $3.548\times10^{-2}$ & $1.976\times10^{-2}$ \\
        & $\mathbf{5}$ & $\mathbf{3}$ & $\mathbf{8.374\times10^{-3}}$ & $\mathbf{4.793\times10^{-3}}$ \\
        & $7$ & $3$ & $1.976\times10^{-3}$ & $1.137\times10^{-3}$ \\
        \hline
        \multirow{8}{*}{Tanh} & $1$ & $1$ & $4.144\times10^{-1}$ & $2.741\times10^{-1}$ \\
        & $3$ & $2$ & $2.322\times10^{-1}$ & $1.423\times10^{-1}$ \\
        & $5$ & $3$ & $1.239\times10^{-1}$ & $7.133\times10^{-2}$ \\
        & $7$ & $3$ & $5.928\times10^{-2}$ & $3.403\times10^{-2}$ \\
        & $9$ & $4$ & $2.686\times10^{-2}$ & $1.593\times10^{-2}$ \\
        & $11$ & $4$ & $1.245\times10^{-2}$ & $7.420\times10^{-3}$ \\
        & $\mathbf{13}$ & $\mathbf{4}$ & $\mathbf{5.924\times10^{-3}}$ & $\mathbf{3.450\times10^{-3}}$ \\
        & $15$ & $4$ & $2.765\times10^{-3}$ & $1.603\times10^{-3}$ \\
        \hline
    \end{tabular}
	}
\end{table}

\begin{figure}[htbp]
    \centering
    \begin{minipage}{\linewidth}
        \centering
        \includegraphics[width=\linewidth]{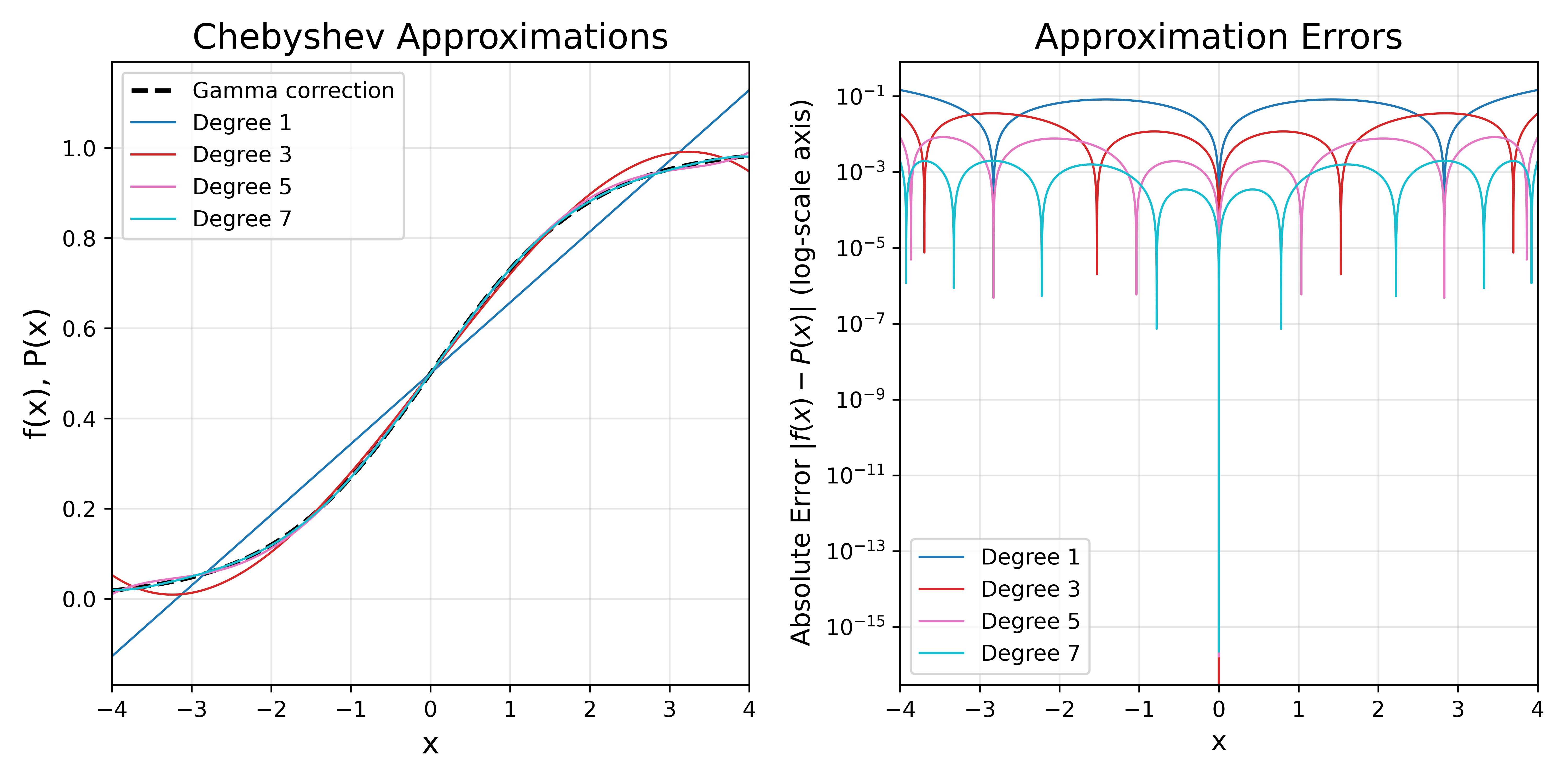}
        \subcaption{Sigmoid.}
    \end{minipage}
    \begin{minipage}{\linewidth}
        \centering
        \includegraphics[width=\linewidth]{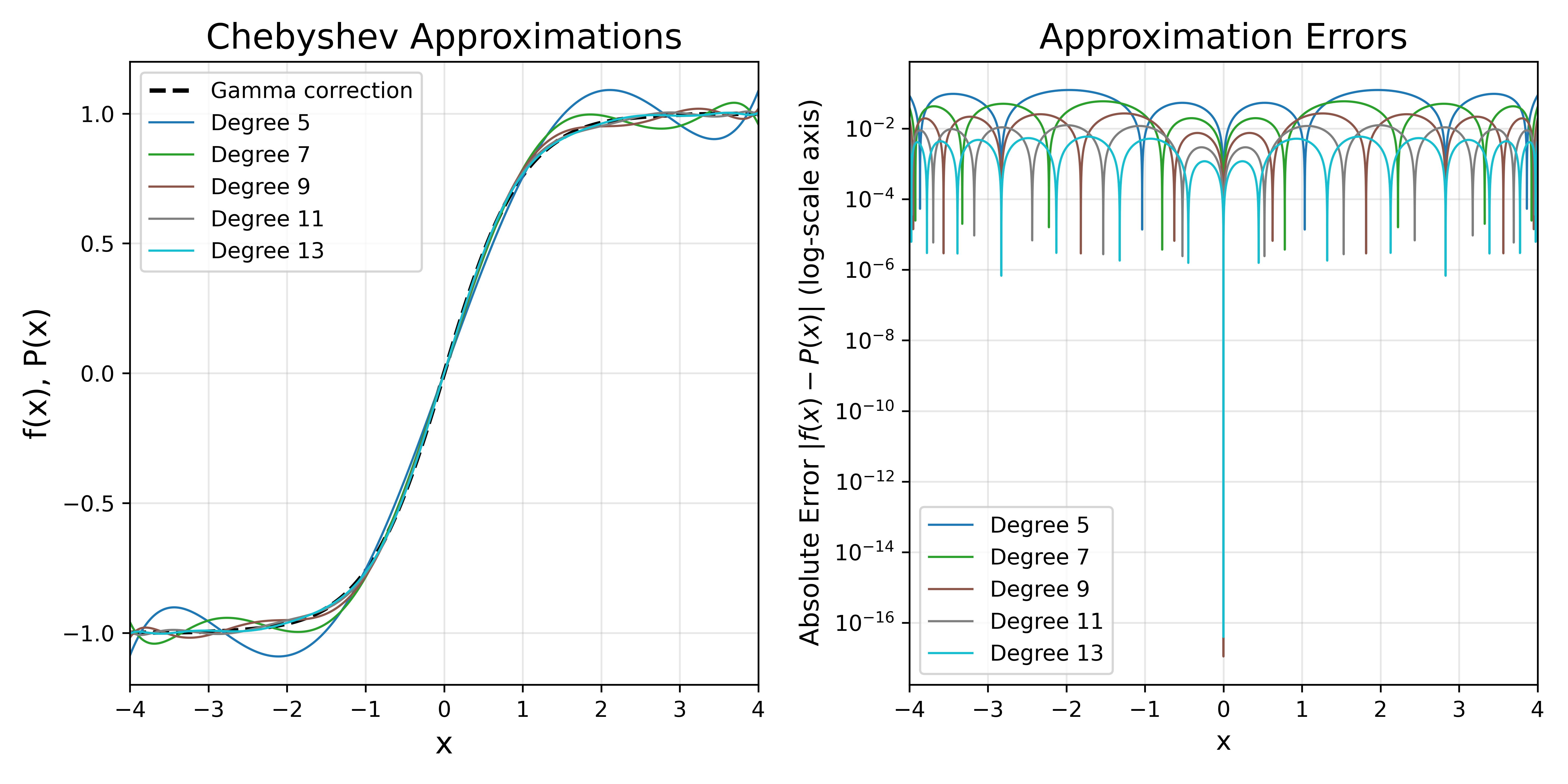}
        \subcaption{Tanh.}
    \end{minipage}
    \caption{Chebyshev approximations of the sigmoid function and the tanh function.}
    \label{fig: chebyshev approximation sigmoid tanh}
\end{figure}

As shown in Table~\ref{tab: errors of sigmoid tanh} and Fig~\ref{fig: chebyshev approximation sigmoid tanh}, targeting $\epsilon_{\rm appro}\leq10^{-2}$, a $5$-degree polynomial suffices for the sigmoid and a $13$-degree polynomial for the tanh.
\begin{gather*}
    \begin{aligned}
        P^{\text{sigmoid}}_5 (x) \approx &~ 0.5 + 2.44647\times10^{-1}x - 1.4269\times10^{-2}x^3 \\
        & + 4.14863\times10^{-4}x^5,
    \end{aligned} \\
    \begin{aligned}
        P^{\text{tanh}}_{13} (x) \approx &~ 9.92421\times10^{-1}x - 2.83839\times10^{-1}x^3  \\
        &+ 6.68517\times10^{-2}x^5 - 9.61612\times10^{-3}x^7 \\
        &+ 7.85724\times10^{-4}x^9 -3.33749\times10^{-5}x^{11} \\
        &+ 5.70796\times10^{-7}x^{13}.
    \end{aligned}
\end{gather*}

In the numerical simulation, we adopt the state-preparation model of~\cite{li2025binary} as a stand-in for the block encoding of $A$ with $\alpha_0=\left\|A\right\|_F$, where the depth and size of implementing are both $\mathcal{O}\left(2^n\right)$. 
This is a common simplification for benchmarking. Since $\left\|A\right\|_F\geq\left\|A\right\|$, this choice satisfies the block-encoding requirement $\alpha_0\geq\left\|A\right\|$.
For the power-oracle constructions, the higher Hadamard powers $A^{\circ 2^l}$ for $l\ge1$ are obtained by classical preprocessing: the matrices $A^{\circ 2^l}$ are computed explicitly and then encoded using the same state-preparation model as $U_A$, with normalization factors $\alpha_l=\|A^{\circ 2^l}\|_F$. 
This is consistent with the power-oracle model defined in Sec.~\ref{subsec: oracle models}, where the block encodings $U_{A^{\circ 2^l}}$ are treated as primitive oracles.
The classical cost of computing $A^{\circ 2^l}$ is not included in the quantum resource counts. 
For the numerical simulation, we randomly sample input states while keeping the sparsity fixed. Since the circuit complexity mainly depends on the system size and the sparsity, this randomness does not affect the comparison of complexities and relative relation of normalization factors across different constructions. 

By implementing these polynomials using the factorization-based, binary-tree-based, and LCU-based QHMF, respectively, we compare the total required circuit depth, size, qubit count, and normalization factor in Fig.~\ref{fig: complexity comparison}. It indicates that: (1) Factorization-based method achieves the lowest depth and size by a substantial margin at the cost of a larger number of ancilla and a higher normalization factor, which verifies previous analysis; (2) Binary‑tree method outperforms the LCU method in both depth and size while maintaining comparable normalization and ancilla overhead.

\begin{figure}[htbp]
    \centering
    \begin{minipage}{\linewidth}
        \centering
        \includegraphics[width=\linewidth]{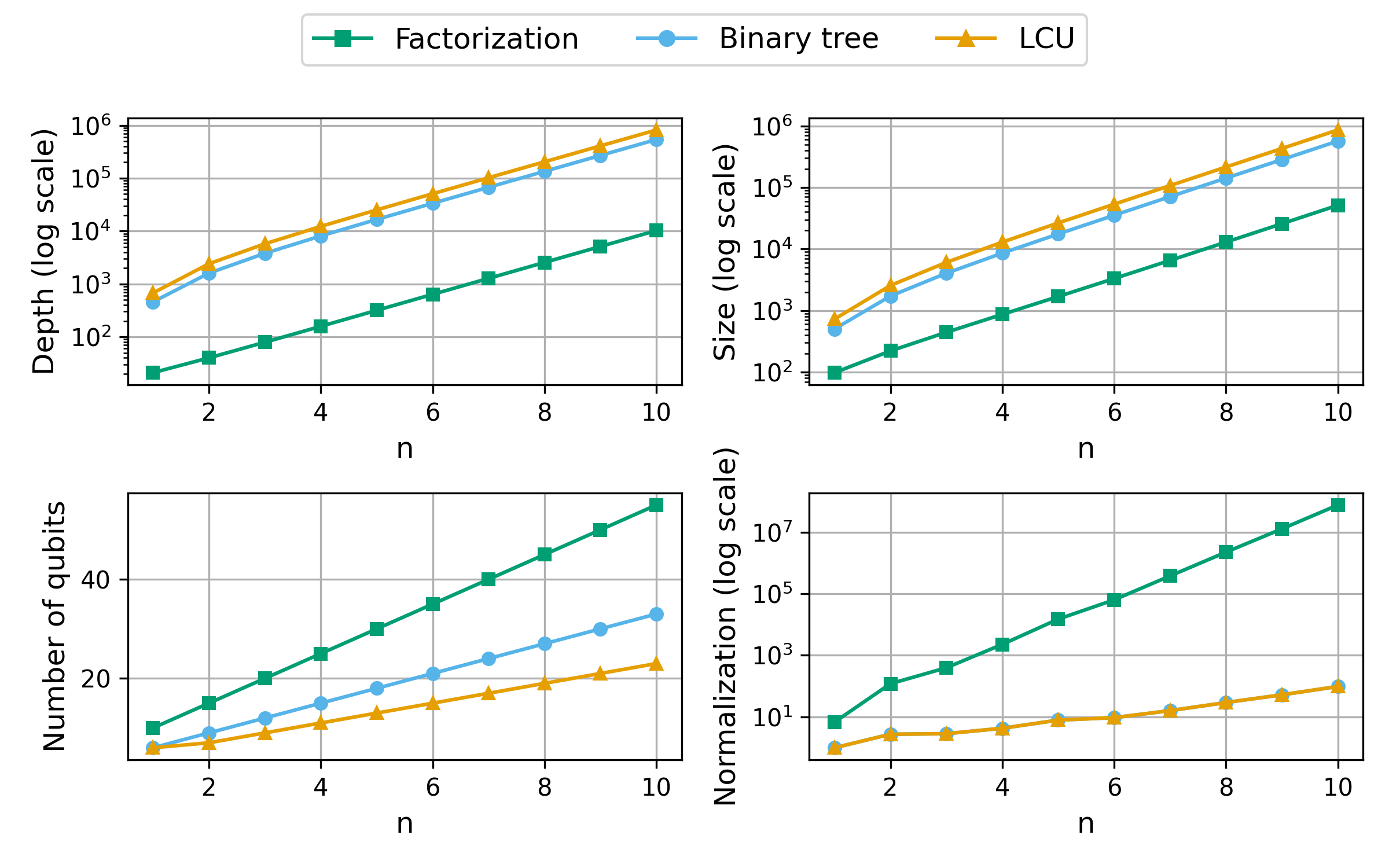}
        \subcaption{$P^{\text{sigmoid}}_5(A)$.}
    \end{minipage}
    \begin{minipage}{\linewidth}
        \centering
        \includegraphics[width=\linewidth]{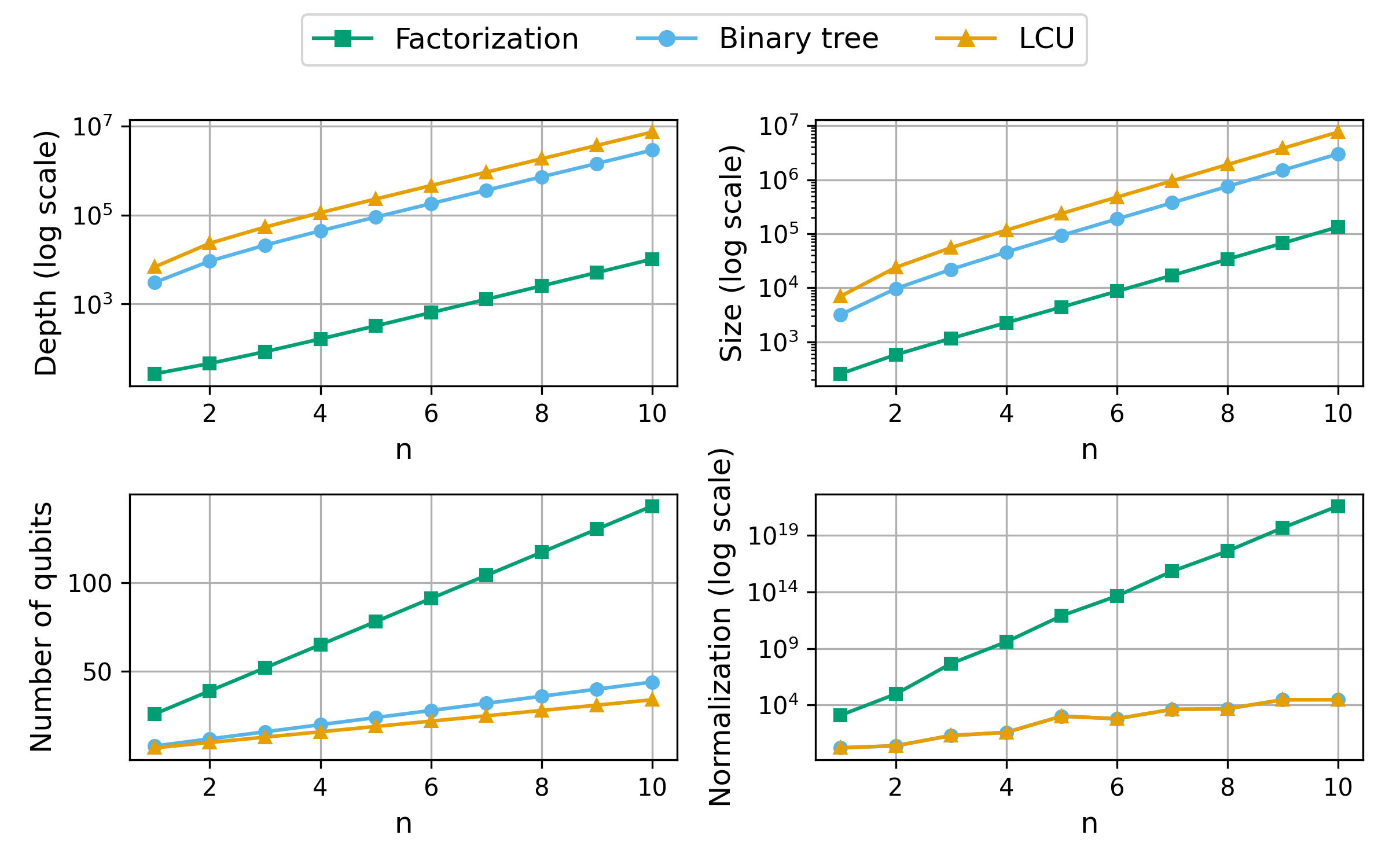}
        \subcaption{$P^{\text{tanh}}_{13}(A)$.}
    \end{minipage}
    \caption{Comparison of circuit complexities for implementing $P^{\text{sigmoid}}_5$ and $P^{\text{tanh}}_{13}$ using the factorization-based, binary-tree-based, and LCU-based~\cite{guo2025quantum} QHMF.}
    \label{fig: complexity comparison}
\end{figure}
\begin{figure}[htbp]
    \centering
    \includegraphics[width=\linewidth]{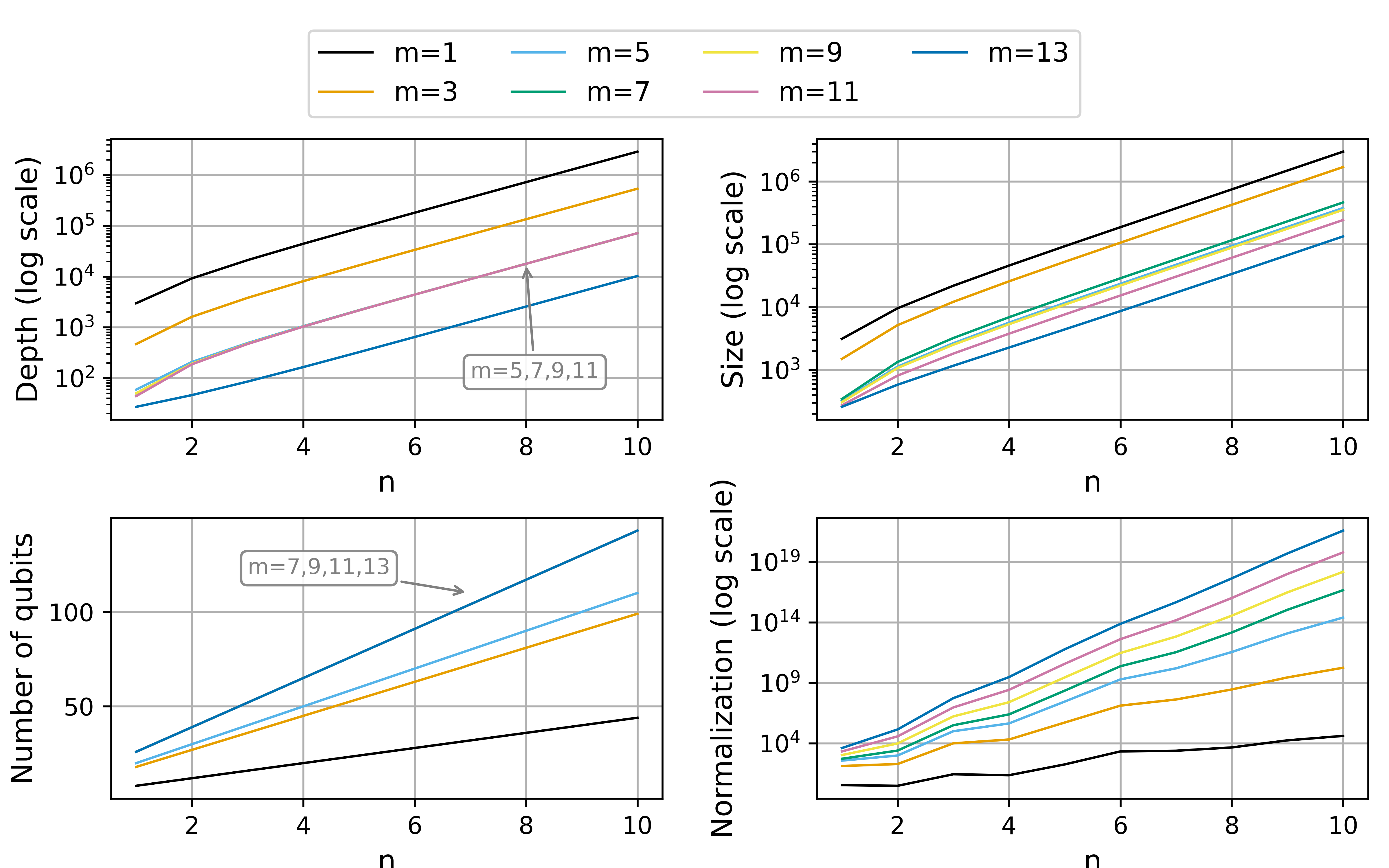}
    \caption{Trade-off complexity for implementing $P^{\text{tanh}}_{13}$ using the trade-off framework.}
    \label{fig: complexity trade-off}
\end{figure}

Moreover, we apply the trade-off QHMF framework to implement the $13$-degree polynomial approximation of the tanh function. The result in Fig.~\ref{fig: complexity trade-off} demonstrates the trade-off between depth and size with the number of qubits and normalization factor that we discuss in Sec.~\ref{sec: trade-off}.

\subsection{Intensity Transformations in Quantum Image Processing}\label{sec: image processing}
In classical image processing, many common operations apply the same mathematical transformation independently to every pixel. Typical examples are the logarithmic transformation and gamma correction, both of which are pointwise functions of the pixel intensity $r\in[0,1]$ and therefore match the entry-wise action of a Hadamard matrix function.

The logarithmic transformation
\begin{equation*}
    s=\log(1+r),
\end{equation*}
expands the dynamic range of low‑intensity regions while compressing high‑intensity regions, thereby revealing fine details in dark areas. Gamma correction is defined by 
\begin{equation*}
    s = r^\gamma,
\end{equation*}
where $\gamma$ is a control parameter. For $\gamma<1$, the mapping stretches dark pixel values, brightening the overall image and enhancing shadow details; for $\gamma>1$, it compresses dark values and expands bright ones, resulting in a darker image with enhanced highlight contrast.

\begin{table}[htbp]
    \centering
    \captionsetup[subtable]{labelformat=parens, labelsep=quad}
    \caption{MaxAE and RMSE for polynomial approximations of the logarithmic transformation and the Gamma correction on $[0,1]$. Errors are evaluated on $20001$ uniformly spaced points. }
    \label{tab: errors of log gamma}
    \begin{subtable}{\linewidth}
        \centering
        \resizebox{\linewidth}{!}{
        \begin{tabular}{|c|c|c|c|c|}
            \hline
            \textbf{Function} & \textbf{Degree} & $\mathbf{d}$ & \textbf{MaxAE} & \textbf{RMSE} \\
            \hline
            \multirow{6}{*}{Log} & $1$ & $1$ & $5.362\times10^{-2}$ & $2.872\times10^{-2}$ \\
            & $2$ & $2$ & $6.308\times10^{-3}$ & $3.373\times10^{-3}$ \\
            & $3$ & $2$ & $8.255\times10^{-4}$ & $4.376\times10^{-4}$ \\
            & $4$ & $3$ & $1.146\times10^{-4}$ & $6.028\times10^{-5}$ \\
            & $\mathbf{5}$ & $\mathbf{3}$ & $\mathbf{1.651\times10^{-5}}$ & $\mathbf{8.635\times10^{-6}}$ \\
            & $6$ & $3$ & $2.443\times10^{-6}$ & $1.271\times10^{-6}$ \\
            \hline
            \multirow{7}{*}{\makecell{Gamma \\ ($\gamma=2.2$)}} & $1$ & $1$ & $1.510\times10^{-1}$ & $9.708\times10^{-2}$ \\
            & $2$ & $2$ & $7.385\times10^{-3}$ & $3.720\times10^{-3}$ \\
            & $3$ & $2$ & $1.268\times10^{-3}$ & $4.774\times10^{-4}$ \\
            & $4$ & $3$ & $3.960\times10^{-4}$ & $1.189\times10^{-4}$ \\
            & $5$ & $3$ & $1.622\times10^{-4}$ & $4.053\times10^{-5}$ \\
            & $\mathbf{6}$ & $\mathbf{3}$ & $\mathbf{7.814\times10^{-5}}$ & $\mathbf{1.672\times10^{-5}}$ \\
            & $7$ & $3$ & $4.201\times10^{-5}$ & $7.860\times10^{-6}$ \\
            \hline
        \end{tabular}
    	}
        \caption{Chebyshev approximations of the logarithmic transformation and the Gamma correction with $\gamma=2.2$. The bold entries are the lowest degrees satisfying $\epsilon_{\rm appro}\leq10^{-4}$.}
    \end{subtable}
    \begin{subtable}{\linewidth}
        \centering
        \resizebox{\linewidth}{!}{
        \begin{tabular}{|c|c|c|c|c|}
            \hline
            \textbf{Method} & \textbf{Degree} & $\mathbf{d}$ & \textbf{MaxAE} & \textbf{RMSE} \\
            \hline
            \multirow{8}{*}{Chebyshev} & $1$ & $1$ & $3.654\times10^{-1}$ & $5.623\times10^{-2}$ \\
            & $2$ & $2$ & $2.562\times10^{-1}$ & $2.643\times10^{-2}$ \\
            & $3$ & $2$ & $2.013\times10^{-1}$ & $1.561\times10^{-2}$ \\
            & $4$ & $3$ & $1.676\times10^{-1}$ & $1.041\times10^{-2}$ \\
            & $\mathbf{5}$ & $\mathbf{3}$ & $\mathbf{1.445\times10^{-1}}$ & $\mathbf{7.492\times10^{-3}}$ \\
            & $6$ & $3$ & $1.275\times10^{-1}$ & $5.679\times10^{-3}$ \\
            & $7$ & $3$ & $1.145\times10^{-1}$ & $4.471\times10^{-3}$ \\
            & $8$ & $4$ & $1.041\times10^{-1}$ & $3.624\times10^{-3}$ \\
            \hline
            \multirow{8}{*}{Minimax} & $1$ & $1$ & $1.629\times10^{-1}$ & $1.108\times10^{-1}$ \\
            & $2$ & $2$ & $9.922\times10^{-2}$ & $6.932\times10^{-2}$ \\
            & $3$ & $2$ & $7.273\times10^{-2}$ & $5.113\times10^{-2}$ \\
            & $4$ & $3$ & $5.807\times10^{-2}$ & $4.093\times10^{-2}$ \\
            & $5$ & $3$ & $4.870\times10^{-2}$ & $3.436\times10^{-2}$ \\
            & $6$ & $3$ & $4.214\times10^{-2}$ & $2.975\times10^{-2}$ \\
            & $7$ & $3$ & $3.728\times10^{-2}$ & $2.633\times10^{-2}$ \\
            & $8$ & $4$ & $3.352\times10^{-2}$ & $2.369\times10^{-2}$ \\
            \hline
        \end{tabular}
		}
        \caption{Chebyshev approximations and minimax approximations of the Gamma correction with $\gamma=0.4$. The bold entry corresponds to the $5$-degree Chebyshev polynomial, which already achieves a PSNR above $48$ dB (computed on the actual pixel distribution) in the image-processing task.}
    \end{subtable}
\end{table}

\begin{figure}[htbp]
    \centering
    \begin{minipage}{\linewidth}
        \centering
        \includegraphics[width=\linewidth]{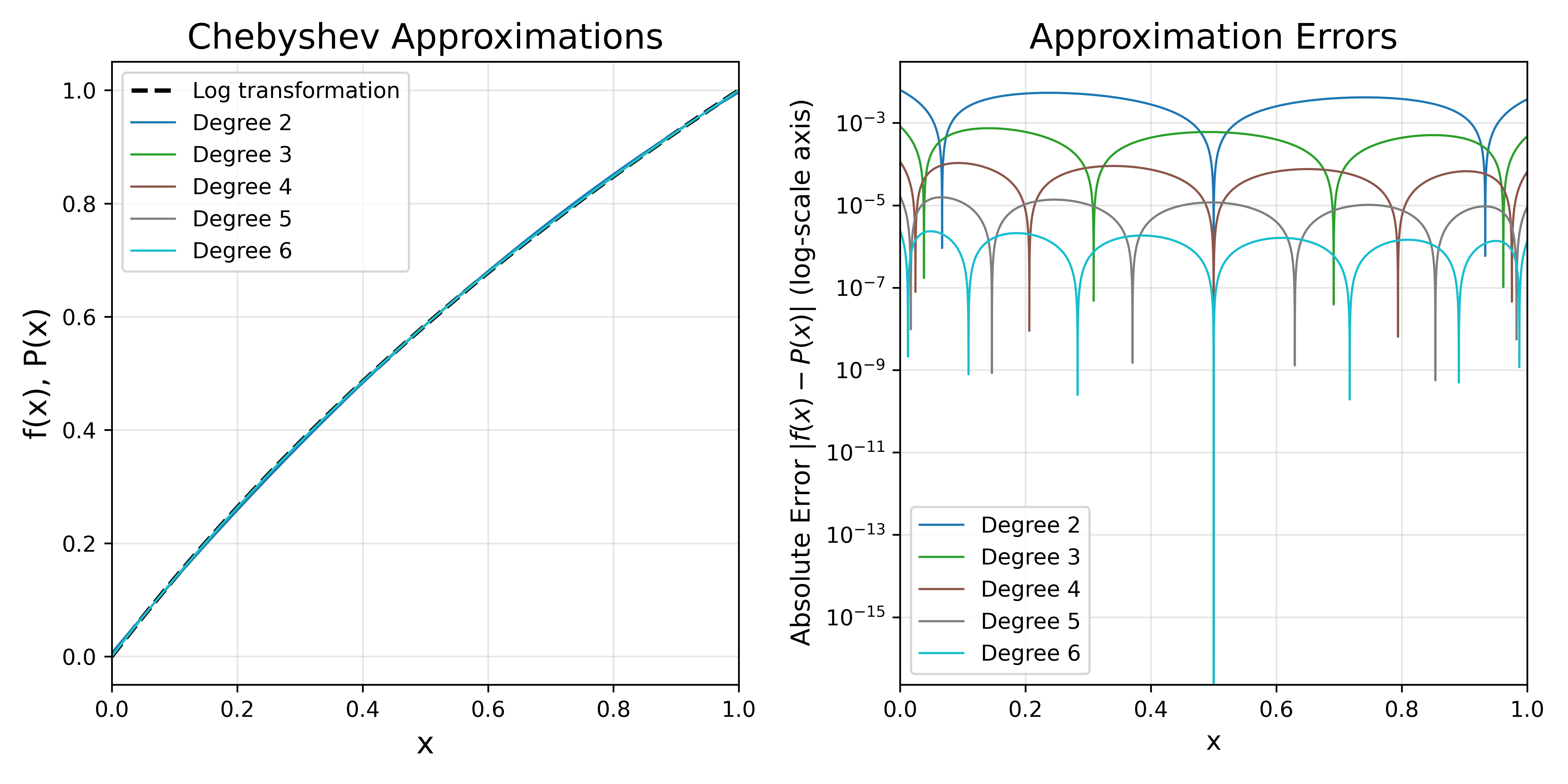}
        \subcaption{Chebyshev approximations of Logarithmic transformation.}
    \end{minipage}
    \begin{minipage}{\linewidth}
        \centering
        \includegraphics[width=\linewidth]{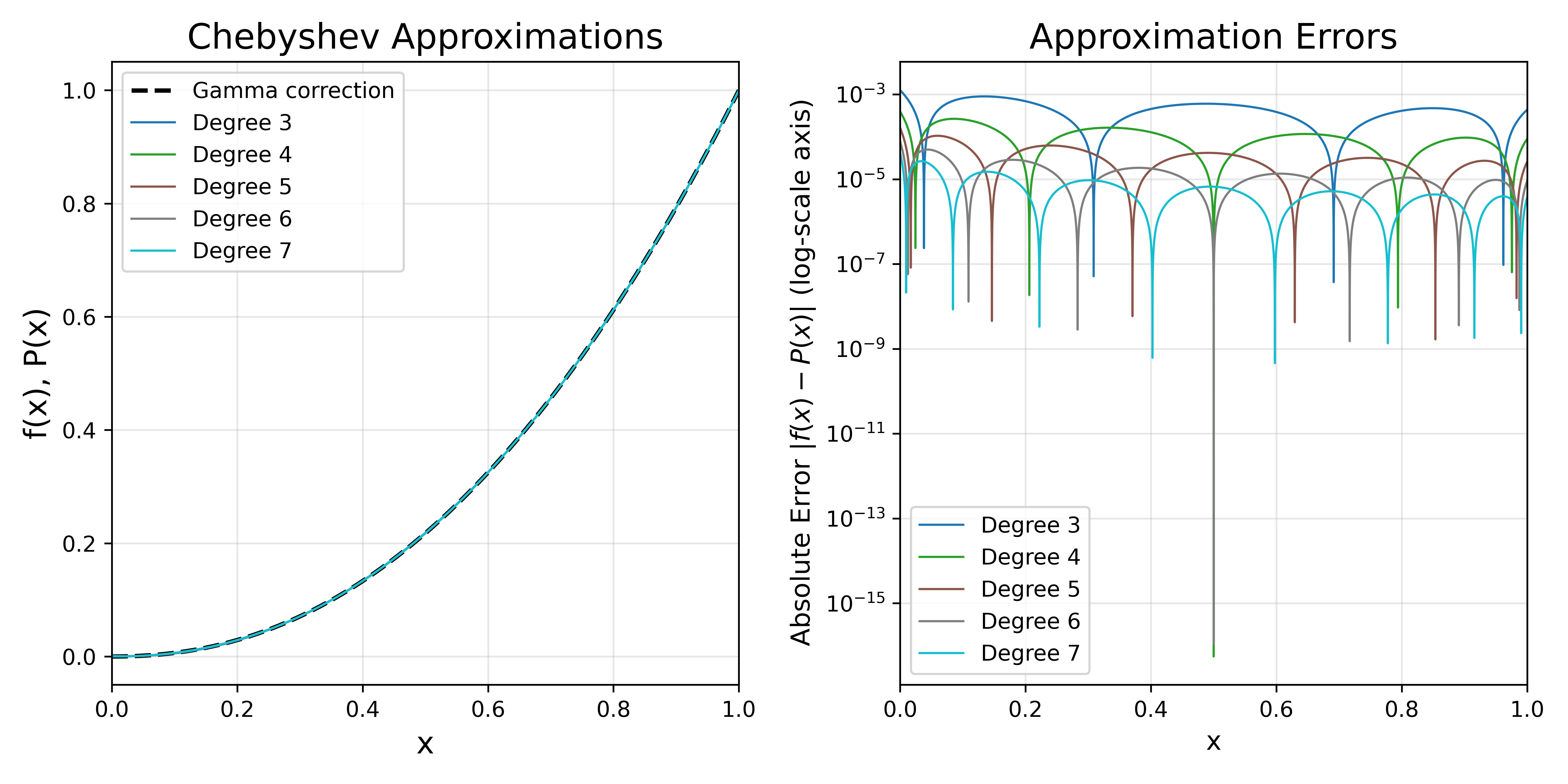}
        \subcaption{Chebyshev approximations of Gamma correction ($\gamma=2.2$).}
    \end{minipage}
    \begin{minipage}{\linewidth}
        \centering
        \includegraphics[width=\linewidth]{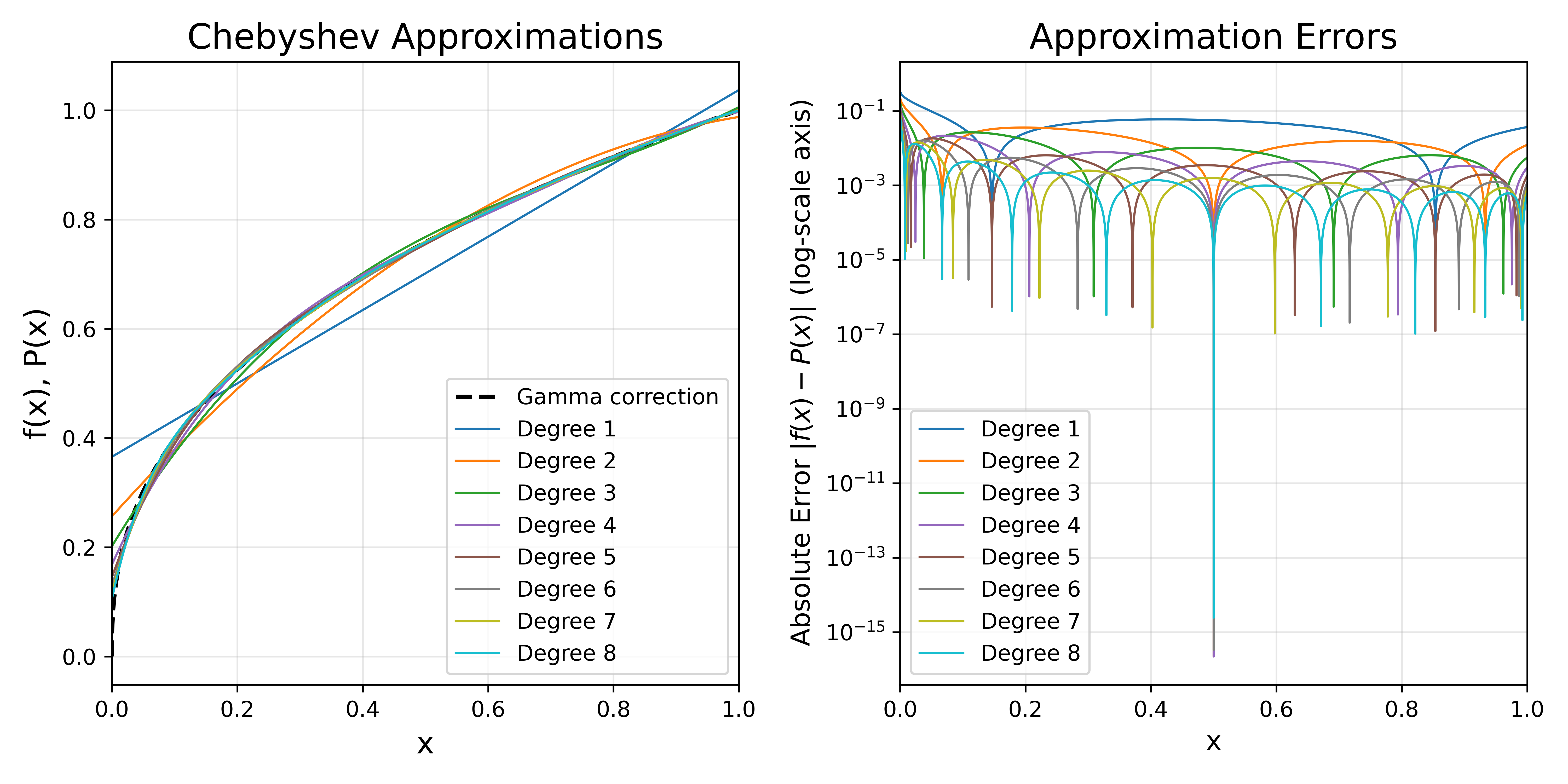}
        \subcaption{Chebyshev approximations of Gamma correction ($\gamma=0.4$).}
        \label{fig: approximation gamma chebyshev}
    \end{minipage}
    \begin{minipage}{\linewidth}
        \centering
        \includegraphics[width=\linewidth]{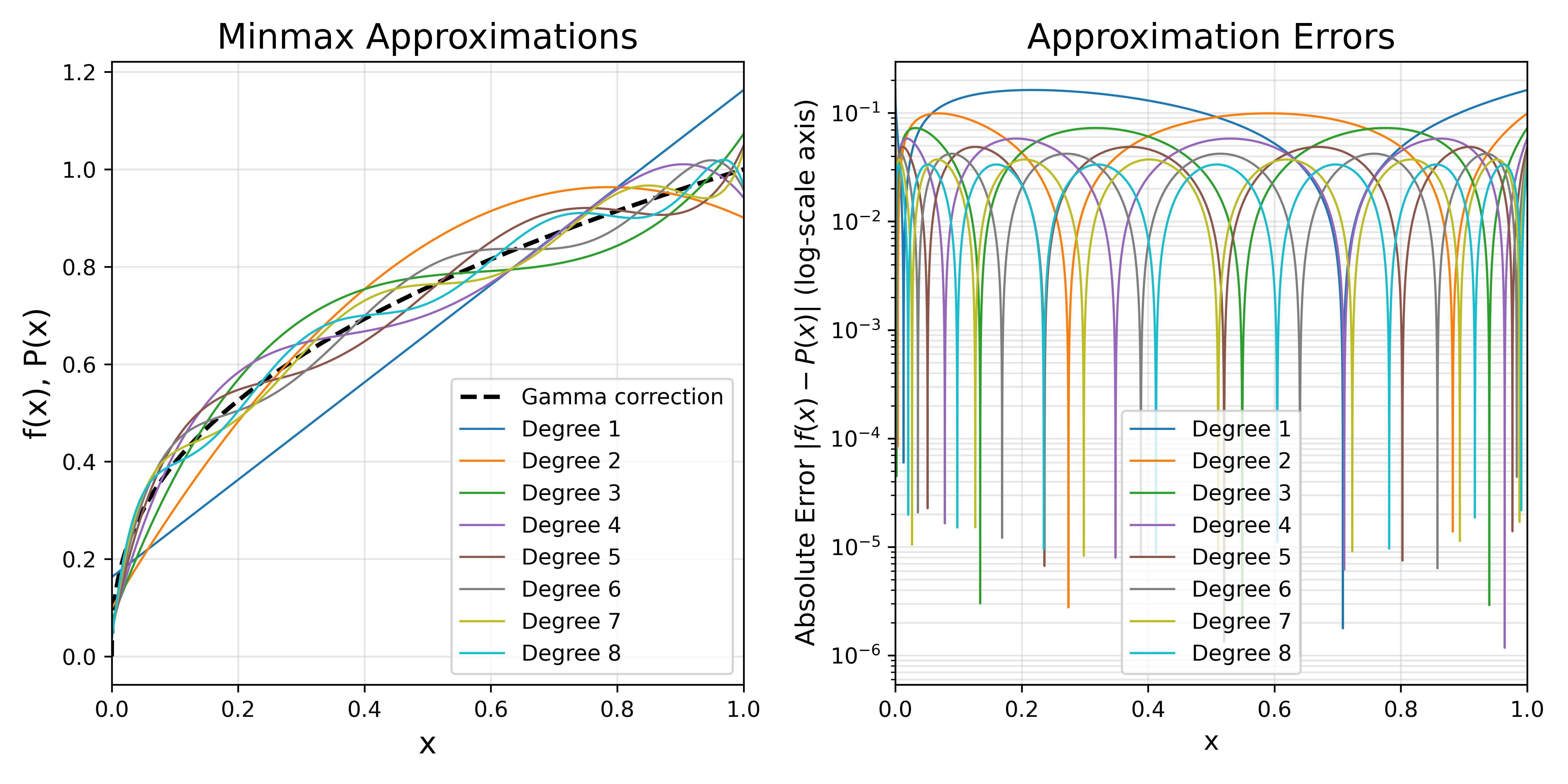}
        \subcaption{Minimax approximations of Gamma correction ($\gamma=0.4$).}
        \label{fig: approximation gamma minimax}
    \end{minipage}
    \caption{Polynomial approximations of the logarithmic transformation and the Gamma correction.}
    \label{fig: approximation log gamma}
\end{figure}

We compute the MaxAE and RMSE of polynomials of varying degrees with power base coefficients obtained by Chebyshev approximations in Table~\ref{tab: errors of log gamma}, and Fig.~\ref{fig: approximation log gamma} presents the visualization of the approximation error. For the logarithmic transformation and the Gamma correction with $\gamma=2.2$, targeting $\epsilon_{\rm appro}\leq 10^{-4}$, a $5$-degree polynomial suffices for the former and a $6$-degree polynomial for the latter. For the Gamma correction with $\gamma=0.4$, since the function is not analytic near the origin, the MaxAE converges slowly, even when using the minimax approximation with the Remez algorithm. Although the minimax polynomial achieves a smaller MaxAE, Fig.~\ref{fig: approximation gamma chebyshev} and~\ref{fig: approximation gamma minimax} show that the error of the Chebyshev approximation is smaller than the error of the minimax approximation with the same degree outside the vicinity of the origin. Through experiments, it is found that the $5$-degree Chebyshev polynomial approximation of the Gamma transformation with $\gamma=0.4$ can also perform well for this image-processing task.
\begin{gather*}
    \begin{aligned}
        P^{\text{log}}_5(x) =& 1.65147\times10^{-5} + 1.44149x - 0.706486 x^2 \\
        & + 0.40947 x^3 -0.187489 x^4 + 0.043005 x^5,
    \end{aligned} \\
    \begin{aligned}
        P^{\gamma=2.2}_6(x) =& 7.81353\times10^{-5} - 9.26559\times10^{-3} x \\
        & + 0.649582 x^2 + 0.714459 x^3 - 0.638808 x^4 \\
        &  + 0.382614 x^5 - 0.0986688 x^6.
    \end{aligned} \\
    \begin{aligned}
        P^{\gamma=0.4}_5(x) =& 0.144471 + 3.17704 x - 8.73113 x^2 \\
        & + 14.8812 x^3 - 12.5205 x^4 + 4.05075 x^5, 
    \end{aligned}
\end{gather*}

\begin{figure}[htbp]
	\centering
	\begin{minipage}[t]{0.49\linewidth}
		\centering
		\includegraphics[width=\linewidth]{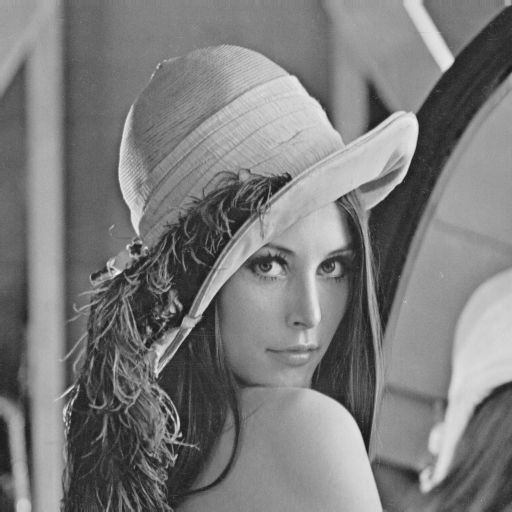}
		\subcaption{Origin image with size $512\times512$.}
		\label{fig: origin image}
	\end{minipage}
	\begin{minipage}[t]{0.49\linewidth}
		\centering
		\includegraphics[width=\linewidth]{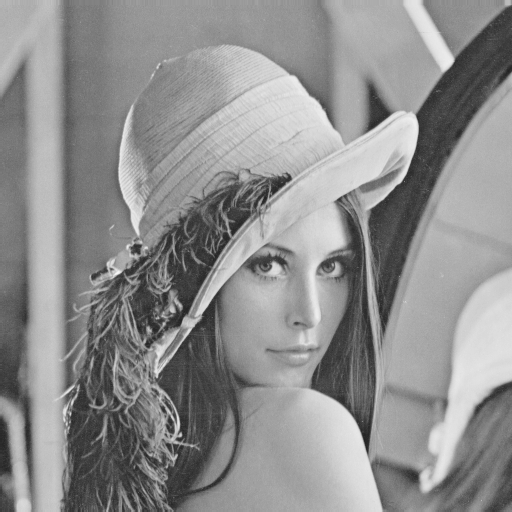}
		\subcaption{Logarithmic‐transformed image; PSNR $= 52.3$ dB.}
		\label{fig: log transformation}
	\end{minipage}
	\begin{minipage}[t]{0.49\linewidth}
		\centering
		\includegraphics[width=\linewidth]{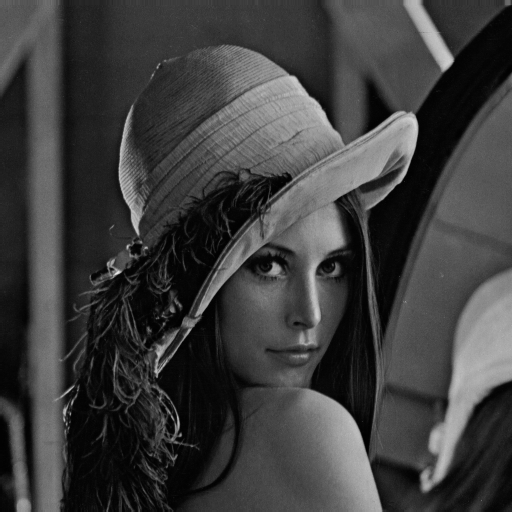}
		\subcaption{Gamma-corrected image with $\gamma=2.2$; PSNR $= 52.8$ dB.}
		\label{fig: gamma correction 2.2}
	\end{minipage}
	\begin{minipage}[t]{0.49\linewidth}
		\centering
		\includegraphics[width=\linewidth]{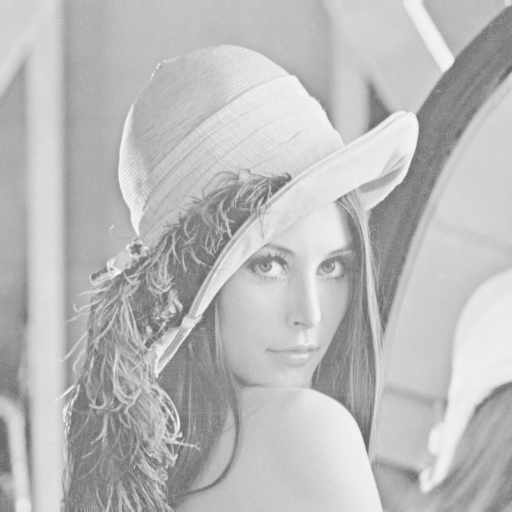}
		\subcaption{Gamma-corrected image with $\gamma=0.4$; PSNR $= 48.9$ dB.}
		\label{fig: gamma correction 0.4}
	\end{minipage}
	\caption{Grayscale image and its log transformation and gamma correction using QHMF. PSNR is computed against the ideal transformation after quantization to $8$-bit grayscale.}
	\label{fig: image and its transformation}
\end{figure}

We apply the approximation polynomials to a grayscale version of ``Lena'' (Fig.~\ref{fig: origin image}) using our QHMF construction. The transformed images are shown in Figs.~\ref{fig: log transformation}-\ref{fig: gamma correction 2.2}. The results exhibit the expected visual effects:
\begin{itemize}
    \item The logarithmic‑transformed image (Fig.~\ref{fig: log transformation}) reveals significantly more detail in the shadow regions while effectively compressing the bright areas, achieving the desired dynamic‑range compression.
    \item The Gamma‑corrected images behave as predicted: with $\gamma=2.2$ (Fig.~\ref{fig: gamma correction 2.2}), the image darkens overall, and contrast shifts toward the highlights; with $\gamma=0.4$ (Fig.~\ref{fig: gamma correction 0.4}), the image becomes distinctly brighter with enhanced contrast in dark regions.
\end{itemize}
In all cases, the polynomial approximations reproduce the ideal transformations to within the average peak signal-to-noise ratio (PSNR) values reported in Fig.~\ref{fig: image and its transformation}. (The PSNR values are computed on the actual pixel distribution of the test image.) This demonstration underscores the versatility of our QHMF framework: it can implement any pointwise polynomial operation on the amplitudes of a quantum state—whether for neural networks, image processing, or other data‑analysis tasks—with resource costs tunable to the underlying hardware.

\section{Conclusion}\label{sec: conclusion}
We have introduced a circuit framework for Hadamard matrix functions that provides a resource-tunable implementation among circuit depth, ancilla overhead, and normalization. 
The framework uses a tunable decomposition of the polynomial controlled by the parameter $m$ to interpolate between two extreme constructions: a logarithmic-depth factorization construction and a binary-tree construction. Its resource scaling is fully characterized: the factorization method for $m=2^d-1$ saturates the query lower bound in the single-oracle model with additional gate depth $\mathcal{O}(n+d)$, while the trade-off constructions for $1\leq m<2^d-1$ offer progressively lower ancilla requirements at the price of increased depth.

A rigorous error analysis complements these resource bounds. It shows that for analytic targets, the required degree grows only logarithmically in the inverse approximation accuracy, and it provides explicit error-propagation bounds for all constructions. These results yield practical guidelines for selecting the construction and the parameter $m$ under a prescribed error budget.
We validated the framework on two representative tasks: Chebyshev approximations of the sigmoid and tanh activation functions, and pointwise logarithmic and gamma transformations of a grayscale image. The numerical results confirmed the predicted depth–ancilla trade-off across the three constructions, and the image transformations achieved high PSNR values.

Several limitations remain. 
First, the trade-off constructions rely on the power-oracle assumption. If the single-oracle model is used, the higher Hadamard powers must be synthesized from $U_A$, which incurs additional query and gate costs and may alter the resource comparison. 
Second, the framework requires classical preprocessing for polynomial decomposition. In the factorization construction, complete decomposition requires root finding, which can become non-trivial for high-degree polynomials. However, this cost is one-time for a given polynomial.
Finally, the shallow-depth extreme suffers from an exponentially small post-selection success probability, which limits its applicability on near-term devices with limited sampling budgets.

\section*{Acknowledgement}
This work is supported by the Fundamental Research Funds for the Central Universities (Grant No. 3072024XX2401).

\section*{Author Contribution Statements}
C.Y.: Conceptualization, Methodology, Investigation, Software, Writing –- original draft. Y.L.: Investigation, Writing –- review \& editing. H.Y.: Conceptualization, Resources, Writing –- review \& editing, Supervision. Z.F.: Writing –- review \& editing, Supervision, Project administration. 

During the preparation of this work, portions of the manuscript were drafted or edited with the assistance of DeepSeek-V4 to improve clarity and style. All authors have reviewed and edited the content as needed and take full responsibility for the content of the publication.

\section*{Code Availability}
We provide Python code that implements our QHMF constructions and was used to generate the results in Sec.~\ref{sec: applications}. The code is openly available at
\url{https://github.com/ChunlinYANG0/QuantumHadamardProcedure}.

\bibliography{references} 

\onecolumn
\appendix

\section{Existing LCU-based Quantum Hadamard matrix function}\label{sec: LCU-based QHMF}
This appendix provides the details of the LCU-based quantum Hadamard matrix function that serves as the baseline in Sec.~\ref{subsec: comparison}. The construction is adapted from Guo et al.~\cite{guo2025quantum}. Although their original formulation was given for a normalized input $P(A/\alpha_0)$, we present it here for the unnormalized polynomial $P(A)$, so that it can be compared with the constructions developed in Sec.~\ref{sec: QHMF} under the same conventions. In the main text, this method is referenced only through its resource counts; here we provide its block-encoding and complexity analysis.

The method treats each Hadamard power $A^{\circ k}$ as an independent term, realizes a block encoding of each term via the Hadamard product of Corollary~\ref{corollary: QHP of m matrices}, and then combines them using an LCU over the polynomial coefficients.
\begin{figure}[htbp]
	\centering
	\begin{tikzpicture}
		\begin{yquant}
			qubit {idx} idx;
			qubit {} anc;
			qubit {$\ket{j}$} j;
			
			["north:$d$" 
			{font=\protect\footnotesize, inner sep=0pt}]
			slash idx;
			["north:$(d-1)n+\sum_{k=0}^{d-1}a_k$" 
			{font=\protect\footnotesize, inner sep=0pt}]
			slash anc;
			["north:$n$" 
			{font=\protect\footnotesize, inner sep=0pt}]
			slash j;
			
			hspace {5pt} -;
			box {$P_L$} idx;
			hspace {5pt} -;
			[multictrl]
			box {$U_{A^{\circ k}}$} (anc,j) ~ idx;
			hspace {5pt} -;
			box {$P_R^{\dagger}$} idx;
			hspace {5pt} -;
			
			output {$\ket{i}$} j;
		\end{yquant}
	\end{tikzpicture}
	\caption{Quantum circuit for the LCU-based QHMF.}
	\label{circuit: LCU-based QHMF}
\end{figure}
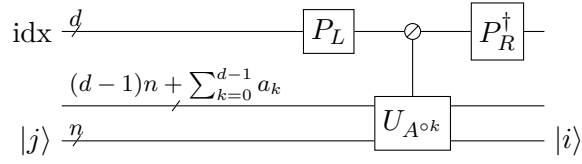
\begin{lemma}[LCU-based QHMF; adapted from~\cite{guo2025quantum}]\label{lem: LCU-based QHMF}
    Let $P(x)=\sum_{k=0}^{2^d-1}c_kx^k$ be a $(2^d-1)$-degree polynomial. In the power-oracle model, assume access to $(\alpha_l,a_l,\epsilon_l)$-block-encodings of an $n$-qubit matrix $A^{\circ2^l}\in\mathbb{C}^{2^n\times 2^n}$ for $l\in\{0,1,\cdots,d-1\}$ together with their controlled versions. Then there exist two state preparation unitaries $P_L$ and $P_R$ satisfying 
	\begin{align*}
		& \begin{aligned}
		    P_L \ket{0}^{\otimes d} 
            = \frac{1}{\sqrt{\alpha_{\rm LCU}}} \left( \sqrt{2^n\left|c_0\right|}\ket{0} + \sum_{k=1}^{2^d-1}\sqrt{\left|c_k\right|\prod_{j=0}^{d-1}\alpha_j^{\operatorname{bit}_j(k)}}\ket{k}\right),
		\end{aligned}
		\\
		& \begin{aligned}
		    P_R \ket{0}^{\otimes d} 
		      = \frac{1}{\sqrt{\alpha_{\rm LCU}}} \left( \sqrt{2^n\left|c_0\right|}e^{-i\operatorname{Arg}\left(c_0\right)}\ket{0} + \sum_{k=1}^{2^d-1}\sqrt{\left|c_k\right| \prod_{j=0}^{d-1}\alpha_j^{\operatorname{bit}_j(k)}}e^{-i\operatorname{Arg} \left(c_k\right)}\ket{k}\right),
		\end{aligned}
	\end{align*}
    where $\alpha_{\rm LCU}=2^n\left|c_0\right| + \sum_{k=1}^{2^d-1}\left|c_k\right|\prod_{j=0}^{d-1}\alpha_j^{\operatorname{bit}_j(k)}$, such that the circuit shown in Fig.~\ref{circuit: LCU-based QHMF} implements $P\left(A\right)$ with the	normalization factor $\alpha_{\rm LCU}$.
\end{lemma}
\begin{proof}
    The matrix that the quantum circuit in Fig.~\ref{circuit: LCU-based QHMF} encodes can be computed as follows:\\
    \begin{equation*}
        \begin{aligned}
            & \left(\bra{0}^{\otimes \left(d+(d-1)n+\sum_{l=0}^{d-1}a_l\right)}\otimes I_{2^{n}}\right) \left(P_R^\dagger\otimes I_{2^{dn+\sum_{l=0}^{d-1}a_l}}\right)  \left(\sum_{k=0}^{2^d-1}\ket{k}\bra{k}\otimes U_{A^{\circ k}}\right) \left(P_L \otimes I_{2^{dn+\sum_{l=0}^{d-1}a_l}}\right) \\
            & \cdot \left(\ket{0}^{\otimes \left(d+(d-1)n+\sum_{l=0}^{d-1}a_l\right)}\otimes I_{2^{n}}\right) \\
            =& \frac{1}{\alpha_{\rm LCU}} \left(\left( \sqrt{2^n\left|c_0\right|}e^{i\operatorname{Arg}\left(c_0\right)}\bra{0} + \sum_{k=1}^{2^d-1}\sqrt{\left|c_k\right| \prod_{j=0}^{d-1}\alpha_j^{\operatorname{bit}_j(k)}}e^{i\operatorname{Arg}\left(c_k\right)}\bra{k}\right)\otimes I_{2^{d(a+n)}}\right)  \\
            & \cdot \left(\ket{0}\bra{0}\otimes \frac{J}{2^n} + \sum_{k=1}^{2^d-1}\ket{k}\bra{k}\otimes \frac{A^{\circ k}}{\prod_{j=0}^{d-1}\alpha_j^{\operatorname{bit}_j(k)}}\right) \left( \left( \sqrt{2^n\left|c_0\right|}\ket{0} + \sum_{k=1}^{2^d-1}\sqrt{\left|c_k\right|\prod_{j=0}^{d-1}\alpha_j^{\operatorname{bit}_j(k)}}\ket{k}\right) \otimes I_{2^n}\right) \\
            =& \frac{1}{\alpha_{\rm LCU}} \sum_{k=0}^{2^d-1} c_k A^{\circ k}.
        \end{aligned}
    \end{equation*}
\end{proof}

The following theorem presents an analysis of the complexity of implementing the LCU-based QHMF.
\begin{lemma}[Complexity of the LCU-based QHMF]\label{lem: complexity of LCU-based QHMF}
	The LCU-based QHMF in Lemma~\ref{lem: LCU-based QHMF} can be implemented with:
    \begin{itemize}
        \item Query complexity: $Q^l_d = \mathcal{O}\left(2^d\right)$.
        \item Additional gates: $D=\mathcal{O}\left(n2^dd\left(\log d\right)^3\right)$, $S=\mathcal{O}\left(n2^dd^2\left(\log d\right)^4\right)$.
        \item Ancilla count: $N=\mathcal{O}\left(nd\right)$.
    \end{itemize}
\end{lemma}
\begin{proof}
	As shown in Fig.~\ref{circuit: LCU-based QHMF}, the circuit for implementing the LCU-based QHMF contains two kinds of oracles:
	\begin{itemize}
		\item[1)] The state preparation oracles $\rm PREP,UNPREP$: \\
		They prepare two $d$-qubit states, which can be implemented with depth and size $\mathcal{O}\left(2^d\right)$~\cite{li2025binary}.
		
		\item[2)] The $d$-qubit controlled block encodings $U_{A^{\circ k}}$ for $k\in\{0,1,\cdots,2^d-1\}$: \\
		\begin{itemize}
			\item[i)] $d$-qubit controlled block-encoding $U_{A^{\circ 0}}$:\\
			There is only the constant term in the LCU that includes $U_{J}$. By Lemma~\ref{lem: block encoding of all-ones matrix}, the circuit size of $U_J$ is $\mathcal{O}\left(n\right)$ using $n$ ancilla. The $d$ control qubits lead that the depth and size for implementing the controlled $U_J$ are $\mathcal{O}\left(n+\left(\log d\right)^3\right)$ and $\mathcal{O}\left(n+d\left(\log d\right)^4\right)$, respectively, using one additional borrowed ancilla.
			
			\item[ii)] $d$-qubit controlled block-encodings $U_{A^{\circ 2^l}}$ for $l\in\{0,1,\cdots,d-1\}$:\\
			For each $l$, the number of $d$-qubit controlled encoding unitaries $U_{A^{\circ 2^l}}$ is $2^{d-1}$. 
			
			\item[iii)] $\text{CX}^{d+1}_1$ gates:\\
			The $\text{CX}^{d+1}_1$ gates are used to implement the QHP to obtain $U_{A^{\circ k}}$ for each $k$. If $k$ is a power of two, then there is no $\text{CX}^{d+1}_1$ gate; otherwise, the number of $\text{CX}^{d+1}_1$ gates is $2n\left(\operatorname{popcount}\left(k\right)-1\right)$, where $\operatorname{popcount}\left(k\right)$ is the number of $1$ in the binary representation of $k$. Thus, the total number of $\text{CX}^{d+1}_1$ gates is 
			\begin{equation*}
				\begin{aligned}
					\sum_{k=0}^{2^d-1}\left(2n\max\left\{\operatorname{popcount}\left(k\right)-1, 0\right\}\right)
					= \mathcal{O}\left(nd2^d\right).
				\end{aligned}
			\end{equation*}
			Because of the control qubits, there are no two of these $\text{CX}^{d+1}_1$ gates that can be implemented in parallel. And a $\text{CX}^{d+1}_1$ gate can be decomposed into a circuit of depth $\mathcal{O}\left(\left(\log d\right)^3\right)$ and size $\mathcal{O}\left(d\left(\log d\right)^4\right)$, using one borrowed ancilla~\cite{claudon2024poly}. 
		\end{itemize}
	\end{itemize}
	Therefore, the total cost of the LCU-based QHMF is
    \begin{itemize}
        \item $\mathcal{O}\left(2^d\right)$ query to each $d$-qubit controlled $U_{A^{\circ 2^l}}$  for $l\in\{0,1,\cdots,d-1\}$
        \item additional gate: 
        \begin{equation*}
            \begin{aligned}
                2\mathcal{O}\left(2^d\right) + \mathcal{O}\left(n+d\left(\log d\right)^4\right) + \mathcal{O}\left(nd2^d\right)\mathcal{O}\left(d\left(\log d\right)^4\right)
                = \mathcal{O}\left(n2^dd^2\left(\log d\right)^4\right)
            \end{aligned}
        \end{equation*}
        with depth 
        \begin{equation*}
            \begin{aligned}
                2\mathcal{O}\left(2^d\right) + \mathcal{O}\left(n+\left(\log d\right)^3\right) + \mathcal{O}\left(nd2^d\right)\mathcal{O}\left(\left(\log d\right)^3\right) 
                = \mathcal{O}\left(n2^dd\left(\log d\right)^3\right)
            \end{aligned}
        \end{equation*}
        \item ancilla: $d+ n+ (d-1)n = \mathcal{O}\left(nd\right)$.
    \end{itemize}
\end{proof}

\section{Proof of Theorems}

\subsection{Proof of Corollary~\ref{corollary: QHP of m matrices}}\label{sec: proof of corollary QHP of m matrices}
\begin{proof}
	Firstly, by performing $GHZ_m^n$ operator and all encoding unitaries $U_{A_k}$, $k\in\{0,1,\cdots,d-1\}$, on the initial state $\left(\bigotimes_{k=1}^{d-1}\left(\ket{0}^{\otimes \tilde{a}_{d-k}}\ket{0}^{\otimes n}\right)\right) \ket{0}^{\otimes \tilde{a}_0}\ket{j}$, we obtain that
	\begin{equation}\label{equation: right part Hadamard product d matrices}
		\begin{aligned}
			\left(\bigotimes_{k=1}^{m-1}\left(\ket{0}^{\otimes \tilde{a}_{m-k}}\ket{0}^{\otimes n}\right)\right) \ket{0}^{\otimes \tilde{a}_0}\ket{j} 
			\xrightarrow{GHZ_m^n} \bigotimes_{k=0}^{m-1} \left(\ket{0}^{\otimes \tilde{a}_{d-1-k}}\ket{j}\right) 
			\xrightarrow{U_{A_k}} \bigotimes_{k=0}^{m-1} \left(U_{A_{m-1-k}}\ket{0}^{\otimes \tilde{a}_{m-1-k}}\ket{j}\right).
		\end{aligned}
	\end{equation}
	Then, by performing $GHZ_m^n$ operator on the state $\left(\bigotimes_{k=1}^{m-1}\left(\ket{0}^{\otimes \tilde{a}_{m-k}}\ket{0}^{\otimes n}\right)\right) \ket{0}^{\otimes \tilde{a}_0}\ket{i}$, we obtain that
	\begin{equation}\label{equation: left part Hadamard product d matrices}
		\begin{aligned}
			\left(\bigotimes_{k=1}^{m-1}\left(\ket{0}^{\otimes \tilde{a}_{m-k}}\ket{0}^{\otimes n}\right)\right) \ket{0}^{\otimes \tilde{a}_0}\ket{i}
			\xrightarrow{GHZ_m^n} \bigotimes_{k=0}^{m-1} \left(\ket{0}^{\otimes \tilde{a}_{m-1-k}}\ket{i}\right).
		\end{aligned}
	\end{equation}
	Finally, taking the inner product of Equations~\eqref{equation: right part Hadamard product d matrices} and~\eqref{equation: left part Hadamard product d matrices}, we have
	\begin{equation*}
		\begin{aligned}
			& \left(\bigotimes_{k=1}^{m-1}\left(\bra{0}^{\otimes \tilde{a}_{m-k}}\bra{0}^{\otimes n}\right)\right) \bra{0}^{\otimes \tilde{a}_0}\bra{i} U_{\bigcirc_{k=0}^{m-1}A_k} \left(\bigotimes_{k=1}^{m-1}\left(\ket{0}^{\otimes \tilde{a}_{m-k}}\ket{0}^{\otimes n}\right)\right) \ket{0}^{\otimes \tilde{a}_0}\ket{j} \\
			=& \left(\bigotimes_{k=0}^{m-1}\left(\bra{0}^{\otimes \tilde{a}_{m-1-k}}\bra{i}\right)\right) \left(\bigotimes_{k=0}^{m-1} \left(U_{A_k}\ket{0}^{\otimes \tilde{a}_{m-1-k}}\ket{j}\right)\right)\\
			=& \frac{1}{\prod_{k=0}^{m-1}\alpha_k} \prod_{k=0}^{m-1}a^{(k)}_{ij}.
		\end{aligned}
	\end{equation*}
	
	Now, we consider the error in the encoding. Let $P$ and $P'$ be the unitary representations of the $GHZ_m^n$ operator that acts on the whole register and the sub-register $\ket{0}^{\otimes (m-1)n}\ket{j}$, respectively. Denote $\bar{A}_k = \tilde{\alpha}_k\left(\bra{0}^{\otimes \tilde{a}_k}\otimes I_{n}\right) U_{A_k} \left(\ket{0}^{\otimes \tilde{a}_k}\otimes I_{n}\right)$, $k\in\{0,1,\cdots,m-1\}$. Since $U_{A_k}$ is an $\left(\tilde{\alpha}_k,\tilde{a}_k,\tilde{\epsilon}_k\right)$-block-encoding of $A_k$, we have $\left\|A_k\right\|\leq\tilde{\alpha}_k$, $\left\|\bar{A}_k\right\|\leq\tilde{\alpha}_k$ and $\left\| A_k - \bar{A}_k\right\|\leq\tilde{\epsilon}_k$. Therefore, the error of the block encoding $U$ can be bounded by
    \begin{equation*}
        \begin{aligned}
            & \left\| \bigcirc_{k=0}^{m-1}A_{k} - \alpha\left(\bra{0}^{\otimes \left((m-1)n+\sum_{k=0}^{m-1}\tilde{a}_k\right)}\otimes I_{n}\right) P  \cdot\left(\bigotimes_{k=0}^{m-1}U_{A_{m-1-k}}\right)P  \left(\ket{0}^{\otimes \left((m-1)n+\sum_{k=0}^{m-1}\tilde{a}_k\right)}\otimes I_{n}\right)\right\| \\
            =& \left\| \bigcirc_{k=0}^{m-1}A_{k} - \left(\bra{0}^{\otimes (m-1)n}\otimes I_{n}\right) P'\left(\bigotimes_{k=0}^{m-1} \bar{A}_{m-1-k}\right) P'\left(\ket{0}^{\otimes (m-1)n}\otimes I_{n}\right) \right\| \\
            =& \left\| \bigcirc_{k=0}^{m-1}A_{k} - \bigcirc_{k=0}^{m-1}\bar{A}_{k} \right\| \\
            =& \left\| \bigcirc_{k=0}^{m-1}A_{k} - \bar{A}_{0} \circ \bigcirc_{k=1}^{m-1}A_{k} + \bar{A}_{0} \circ \bigcirc_{k=1}^{m-1}A_{k} - \cdots + \bigcirc_{k=0}^{m-2}\bar{A}_{k} \circ A_{m-1} - \bigcirc_{k=0}^{m-1}\bar{A}_{k} \right\| \\
            \leq& \left\| \left(A_{0} - \bar{A}_{0}\right) \circ \bigcirc_{k=1}^{m-1}A_{k} \right\| + \cdots + \left\| \bigcirc_{k=0}^{m-2}\bar{A}_{k} \circ \left( A_{m-1} - \bar{A}_{m-1}\right) \right\| \\
            \leq& \sum_{k=0}^{m-1} \left(\frac{\tilde{\epsilon}_k}{\tilde{\alpha}_k} \prod_{l=0}^{m-1}\tilde{\alpha}_{l}\right).
        \end{aligned}
    \end{equation*}
\end{proof}

\subsection{Proof of Theorem~\ref{thm: factorization-based QHMF}}\label{sec: proof of factorization-based QHMF}
\begin{proof}
	First, we perform the $GHZ_{2^d-1}^n$ operator, rotation gates $R_Y\left(\theta_k\right),R_Z\left(\phi_k\right)$ for $k\in\{0,1,\cdots,2^d-2\}$, single-qubit controlled block encodings $U_A$ and $U_{J}$ on the initial state $\ket{0}^{\otimes (1+\hat{a}_0+n)\left(2^d-2\right)} \ket{0}^{\otimes (1+\hat{a}_0)}\ket{j}$ successively, which results in
	\begin{equation}\label{equation: right part Hadamard matrix function function}
		\begin{aligned}
			& \ket{0}^{\otimes (1+\hat{a}_0+n)\left(2^d-2\right)} \ket{0}^{\otimes (1+\hat{a}_0)}\ket{j} \\
			\xrightarrow{GHZ_{2^d-1}^n} & \left(\ket{0}^{\otimes (1+\hat{a}_0)}\ket{j}\right)^{\otimes \left(2^d-1\right)} \\
			\xrightarrow{R_Y\left(\theta_k\right),R_Z\left(\phi_k\right)} & \bigotimes_{k=0}^{2^d-2}\left(e^{-i\frac{\phi_k}{2}}\cos\frac{\theta_{k}}{2}\ket{0}\ket{0}^{\otimes \hat{a}_0}\ket{j}+e^{i\frac{\phi_k}{2}}\sin\frac{\theta_{k}}{2}\ket{1}\ket{0}^{\otimes \hat{a}_0}\ket{j}\right) \\
			\xrightarrow{U_A,U_{J}}& \bigotimes_{k=0}^{2^d-2}\left(e^{-i\frac{\phi_{k}}{2}}\cos\frac{\theta_{k}}{2} \ket{0}\left(U_A\ket{0}^{\otimes \hat{a}_0}\ket{j}\right) +e^{i\frac{\phi_{k}}{2}}\sin\frac{\theta_{k}}{2}\ket{1}\left(U_{J}\ket{0}^{\otimes \hat{a}_0}\ket{j}\right)\right).
		\end{aligned}
	\end{equation}
	
	Then, by performing the $GHZ_{2^d-1}^n$ operator, rotation gate $R_Z^\dagger\left(\varphi\right)$, and rotation gates $R_Y^{\dagger}\left(\theta_k\right)$ for $k\in\{0,1,\cdots,2^d-2\}$ on the state $\ket{0}^{\otimes (1+\hat{a}_0+n)\left(2^d-2\right)} \ket{0}^{\otimes (1+\hat{a}_0)}\ket{i}$, we obtain that
	\begin{equation}\label{equation: left part Hadamard matrix function function}
		\begin{aligned}
			& \ket{0}^{\otimes (1+\hat{a}_0+n)\left(2^d-2\right)} \ket{0}^{\otimes (1+\hat{a}_0)}\ket{i} \\
			\xrightarrow{GHZ_{2^d-1}^n} & \left(\ket{0}^{\otimes (1+\hat{a}_0)}\ket{i}\right)^{\otimes \left(2^d-1\right)} \\
			\xrightarrow{R_Z^\dagger\left(\varphi\right)} & e^{i\frac{\varphi}{2}}\left(\ket{0}^{\otimes (1+\hat{a}_0)}\ket{i}\right)^{\otimes \left(2^d-1\right)} \\
			\xrightarrow{R_Y^{\dagger}\left(\theta_k\right)}& e^{i\frac{\varphi}{2}}\bigotimes_{k=0}^{2^d-2} \left(\cos\frac{\theta_{k}}{2}\ket{0}\ket{0}^{\otimes \hat{a}_0}\ket{i} - \sin\frac{\theta_{k}}{2}\ket{1}\ket{0}^{\otimes \hat{a}_0}\ket{i}\right). 
		\end{aligned}
	\end{equation}
	
	Finally, taking the inner product between Eq.~\eqref{equation: left part Hadamard matrix function function} and~\eqref{equation: right part Hadamard matrix function function}, we have
	\begin{equation*}
		\begin{aligned}
			& \bra{0}^{\otimes (1+\hat{a}_0+n)\left(2^d-2\right)} \bra{0}^{\otimes (1+\hat{a}_0)}\bra{i} U_{P(A)}
			\ket{0}^{\otimes (1+\hat{a}_0+n)\left(2^d-2\right)} \ket{0}^{\otimes (1+\hat{a}_0)}\ket{j} \\
			=& \left(e^{-i\frac{\varphi}{2}}\bigotimes_{k=0}^{2^d-2}\left(\cos\frac{\theta_{k}}{2}\bra{0}\bra{0}^{\otimes \hat{a}_0}\bra{i} -\sin\frac{\theta_{k}}{2}\bra{1}\bra{0}^{\otimes \hat{a}_0}\bra{i}\right)\right) \\
			& \cdot \left(\bigotimes_{k=0}^{2^d-2}\left(e^{-i\frac{\phi_{k}}{2}}\cos\frac{\theta_{k}}{2}\ket{0}\left(U_A\ket{0}^{\otimes a}\ket{j}\right) +e^{i\frac{\phi_{k}}{2}}\sin\frac{\theta_{k}}{2}\ket{1}\left(U_{J}\ket{0}^{\otimes \hat{a}_0}\ket{j}\right)\right)\right) \\
			=& e^{-i\frac{\varphi}{2}}\prod_{k=0}^{2^d-2} \left(\left(e^{-i\frac{\phi_k}{2}}\cos^2\frac{\theta_k}{2}\right) \frac{a_{ij}}{\alpha_0} - \left(e^{i\frac{\phi_k}{2}}\sin^2\frac{\theta_k}{2}\right) \frac{1}{2^n}\right) \\
			=& \frac{1}{\left|c_{2^d-1}\right|}e^{-i\left(\frac{\varphi}{2}+\sum_{k=0}^{2^d-2}\frac{\phi_k}{2} + \operatorname{Arg}\left(c_{2^d-1}\right)\right)} \left(\prod_{k=0}^{2^d-2}\frac{\cos^2\frac{\theta_k}{2}}{\alpha_0}\right) c_{2^d-1} \prod_{k=0}^{2^d-2} \left( a_{ij} - e^{i\phi_k}\frac{\alpha_0}{2^n}\tan^2\frac{\theta_k}{2} \right) \\
            =& \frac{1}{\left|c_{2^d-1}\right|\prod_{k=0}^{2^d-2}\left(\alpha_0+2^n\left|r_k\right|\right)} P(a_{ij}).
		\end{aligned}
	\end{equation*}	
\end{proof}

\subsection{Proof of Theorem~\ref{theorem: binary-tree-based QHMF}}\label{sec: proof for binary-tree-based QHMF}
\begin{proof}   
    The matrix that the circuit in Fig.~\ref{circuit: 1-degree QHMF} encodes can be computed as follows:
    \begin{equation*}
        \begin{aligned}
            &  \left(\bra{0}\otimes \bra{0}^{\otimes a_0} \otimes I_{2^n}\right) U_{P_{\left[2k_{d-1},2k_{d-1}+1\right]}\left(A\right)} \left(\ket{0}\otimes \ket{0}^{\otimes a_0} \otimes I_{2^n}\right) \\
            =& \left(\bra{0}\otimes \bra{0}^{\otimes a_0} \otimes I_{2^n}\right) \left(R_Y\left(-\theta_{k_{d-1}}\right)R_Z\left(\varphi_{k_{d-1}}\right)\otimes I_{2^{a_0+n}}\right) \left(\ket{0}\bra{0}\otimes U_{J} + \ket{1}\bra{1}\otimes U_{A}\right) \\
            & \left(R_Y\left(\theta_{k_{d-1}}\right)R_Z\left(\phi_{k_{d-1}}\right)\otimes I_{2^{a_0+n}}\right) \left(\ket{0}\otimes \ket{0}^{\otimes a_0} \otimes I_{2^n}\right) \\
            =& \left(\left(e^{-i\frac{\varphi_{k_{d-1}}}{2}}\cos\frac{\theta_{k_{d-1}}}{2}\bra{0} + e^{i\frac{\varphi_{k_{d-1}}}{2}}\sin\frac{\theta_{k_{d-1}}}{2}\bra{1}\right) \otimes I_{2^n} \right) \left(\ket{0}\bra{0}\otimes \frac{J}{2^n} + \ket{1}\bra{1}\otimes \frac{A}{\alpha_0}\right) \\
            & \left(\left(e^{-i\frac{\phi_{k_{d-1}}}{2}}\cos\frac{\theta_{k_{d-1}}}{2}\ket{0} + e^{-i\frac{\phi_{k_{d-1}}}{2}}\sin\frac{\theta_{k_{d-1}}}{2}\ket{1}\right) \otimes I_{2^n} \right) \\
            =& \left( e^{-i\frac{\phi_{k_{d-1}}+\varphi_{k_{d-1}}}{2}}\cos^2\frac{\theta_{k_{d-1}}}{2} \right) \frac{J}{2^n} + \left(e^{-i\frac{\phi_{k_{d-1}}-\varphi_{k_{d-1}}}{2}}\sin^2\frac{\theta_{k_{d-1}}}{2}\right) \frac{A}{\alpha_0} \\
            =& \frac{1}{2^n\left|c_{2k_{d-1}}\right| + \left|c_{2k_{d-1}+1}\right|\alpha_0} \left( c_{2k_{d-1}} J + c_{2k_{d-1}+1}A\right) \\
            =& \frac{1}{\alpha_{d-1,k_{d-1}}} P_{\left[2k_{d-1},2k_{d-1}+1\right]}\left(A\right).
        \end{aligned}
    \end{equation*}
    
    The matrix that the recursive circuit in Fig.~\ref{circuit: recursive quantum circuit} encodes can be computed as follows:
    \begin{equation*}
        \begin{aligned}
            & \left(\bra{0}^{\otimes \left(1+a_{d-l-1}+n+b_l\right)}\otimes I_{2^n}\right) U_{P_{\left[2^{d-l}k_{l},2^{d-l}(k_{l}+1)-1\right]}\left(A\right)} \left(\ket{0}^{\otimes \left(1+a_{d-l-1}+n+b_l\right)}\otimes I_{2^n}\right) \\
            =& \left(\bra{0}^{\otimes \left(1+a_{d-l-1}+n+b_l\right)}\otimes I_{2^n}\right) \left(R_Y\left(-\gamma_{l,k_{l}}\right) \otimes I_{2^{a_{d-l-1}+n+b_l+n}}\right) \left(\ket{0}\bra{0}\otimes I_{2^{a_{d-l-1}+n}} \otimes U_{P_{\left[2^{d-l}k_{l},2^{d-l-1}(2k_{l}+1)-1\right]}\left(A\right)}\right. \\
            & \left.+ \ket{1}\bra{1}\otimes\left( {\rm CNOT}^{\otimes n} \left( U_{A^{\circ2^{d-l-1}}} \otimes U_{P_{\left[2^{d-l-1}(2k_{l}+1),2^{d-l}(k_{l}+1)-1\right]}\left(A\right)} \right) {\rm CNOT}^{\otimes n} \right) \right) \\
            & \left(R_Y\left(\gamma_{l,k_{l}}\right) \otimes I_{2^{a_{d-l-1}+n+b_{l}+n}}\right) \left(\ket{0}^{\otimes \left(1+a_{d-l-1}+n+b_l\right)}\otimes I_{2^n}\right) \\
            =& \left(\left(\cos\frac{\gamma_{l,k_{l}}}{2}\bra{0} + \sin\frac{\gamma_{l,k_{l}}}{2}\bra{1}\right) \otimes I_{2^{n}}\right) \left(\ket{0}\bra{0}\otimes \frac{P_{\left[2^{d-l}k_{l},2^{d-l-1}(2k_{l}+1)-1\right]}\left(A\right)}{\alpha_{l+1,2k_l}} \right. \\
            & \left. + \ket{1}\bra{1}\otimes \left(\frac{A^{\circ2^{d-l-1}}}{\alpha_{d-l-1}} \circ \frac{P_{\left[2^{d-l-1}(2k_{l}+1),2^{d-l}(k_{l}+1)-1\right]}\left(A\right)}{\alpha_{l+1,2k_l+1}} \right) \right) \left(\left(\cos\frac{\gamma_{l,k_{l}}}{2}\ket{0} + \sin\frac{\gamma_{l,k_{l}}}{2}\ket{1}\right) \otimes I_{2^{n}}\right) \\
            =& \left(\cos^2\frac{\gamma_{l,k_{l}}}{2}\right) \frac{P_{\left[2^{d-l}k_{l},2^{d-l-1}(2k_{l}+1)-1\right]}\left(A\right)}{\alpha_{l+1,2k_l}} + \left(\sin^2\frac{\gamma_{l-1,k_{l-1}}}{2}\right) \frac{P_{\left[2^{d-l-1}(2k_{l}+1),2^{d-l}(k_{l}+1)-1\right]}\left(A\right)}{\alpha_{l+1,2k_l+1}} \circ \frac{A^{\circ 2^{d-l-1}}}{\alpha_{d-l-1}} \\
            =& \frac{P_{\left[2^{d-l}k_{l},2^{d-l-1}(2k_{l}+1)-1\right]}\left(A\right) + P_{\left[2^{d-l-1}(2k_{l}+1),2^{d-l}(k_{l}+1)-1\right]}\left(A\right) \circ A^{\circ 2^{d-l-1}}}{\alpha_{l+1,2k_l} + \alpha_{d-l-1} \alpha_{l+1,2k_l+1}} \\
            =& \frac{1}{\alpha_{l,k_l}} P_{\left[2^{d-l}k_{l},2^{d-l}(k_{l}+1)-1\right]}\left(A\right),
        \end{aligned}
    \end{equation*}
    where the last equality holds because $\alpha_{l, k_l} = \alpha_{l+1,2k_l} + \alpha_{d-l-1} \alpha_{l+1,2k_l+1}$.
    
    By inductive method, the circuit of $U_{P\left(A\right)}$ implements the Hadamard matrix function $P\left(A\right)$ with a normalization factor $\alpha_{0,0} = \sum_{k=0}^{2^d-1} 2^{n\overline{\operatorname{bit}_0(k)}}\alpha_0^{\operatorname{bit}_0(k)} |c_k|\prod_{j=1}^{d-1}\alpha_j^{\operatorname{bit}_j(k)}$.
\end{proof}

\subsection{Proof of Corollary~\ref{cor: complexity of binary-tree-based QHMF}}\label{sec: proof of complexity for binary-tree QHMF}
\begin{proof}
	The circuit of the binary-tree-based QHMF in Fig.~\ref{circuit: binary-tree-based QHMF} contains four kinds of gates: 
	\begin{itemize}
		\item[1)] (Controlled) $R_Y$ rotations for $k_l\in\{0,1,\cdots,2^l-1\}$ and $l\in\{0,1,\cdots,d-2\}$: \\
		In the $l$-th tree layer, there are $2^{l+1}$ $l$-qubit controlled $R_Y$ rotations. And an $l$-qubit controlled $R_Y$ can be decomposed into a circuit with depth $\mathcal{O}\left(\left(\log l\right)^3\right)$ and size $\mathcal{O}\left(l\left(\log l\right)^4\right)$, using one borrowed ancilla~\cite{claudon2024poly}.
		
		\item[2)] Multi-controlled $X$ gates: \\
		In the $l$-th tree layer for $l\in\{0,1,\cdots,d-2\}$, there are $n2^{l+1}$ $\text{CX}^{l+2}_1$ gates. And a $\text{CX}^{l+2}_1$ gate can be decomposed into a circuit with depth $\mathcal{O}\left(\left(\log l\right)^3\right)$ and size  $\mathcal{O}\left(l\left(\log l\right)^4\right)$, using one borrowed ancilla~\cite{claudon2024poly}.
		
		\item[3)] Controlled $R_Z,R_Y$ rotations in the $(d-1)$-th recursive layer:\\
		These controlled rotations are controlled by $(d-1)$ qubits. The number of controlled rotations is $2^{d+1}$. And a $(d-1)$-qubit controlled rotation gate can be decomposed into a circuit with depth $\mathcal{O}\left(\left(\log d\right)^3\right)$ and size $\mathcal{O}\left(d\left(\log d\right)^4\right)$, using one borrowed ancilla~\cite{claudon2024poly}.
		
		\item[4)] Controlled block encodings $U_{J}$ and $U_{A^{\circ 2^l}}$ for $l\in\{0,1,\cdots,d-1\}$:\\
		In the $l$-th tree layer, there are $2^l$ $(l+1)$-qubit controlled block encodings $U_{A^{\circ 2^{d-l-1}}}$ for $l\in\{0,1,\cdots,d-1\}$. In the $(d-1)$-th layer, there are $2^{d-1}$ $d$-qubit controlled block encodings $U_{J}$. By Lemma~\ref{lem: block encoding of all-ones matrix} and Ref.~\cite{claudon2024poly}, a $d$-qubit controlled $U_{J}$ can be implemented with depth $\mathcal{O}\left(n+\left(\log d\right)^3\right)$ and size $\mathcal{O}\left(n+d\left(\log d\right)^4\right)$, using one additional borrowed ancilla.
	\end{itemize}
	Therefore, the binary-tree-based QHMF can be implemented with cost:
    \begin{itemize}
        \item $\mathcal{O}\left(2^{d-l}\right)$ query to each $(d-l)$-qubit controlled block encodings $U_{A^{\circ 2^{l}}}$ for $l\in\{0,1,\cdots,d-1\}$,
        \item additional gate:
        \begin{equation*}
            \begin{aligned}
                & \sum_{l=1}^{d-2}2^{l+1}\mathcal{O}\left(l\left(\log l\right)^4\right) + \sum_{l=1}^{d-2}n2^{l+1}\mathcal{O}\left(l\left(\log l\right)^4\right) + 2^{d+1}\mathcal{O}\left(d\left(\log d\right)^4\right) + 2^{d-1}\mathcal{O}\left(n+d\left(\log d\right)^4\right) \\
                =& \mathcal{O}\left(n2^dd\left(\log d\right)^4\right)
            \end{aligned}
        \end{equation*}
        with depth
        \begin{equation*}
            \begin{aligned}
                & \sum_{l=1}^{d-2}2^{l+1}\mathcal{O}\left(\left(\log l\right)^3\right) + \sum_{l=1}^{d-2}n2^{l+1}\mathcal{O}\left(\left(\log l\right)^3\right) + 2^{d+1}\mathcal{O}\left(\left(\log d\right)^3\right) + 2^{d-1}\mathcal{O}\left(n +\left(\log d\right)^3\right) \\
                =& \mathcal{O}\left(n2^d\left(\log d\right)^3\right)
            \end{aligned}
        \end{equation*}
        \item ancilla: $d+n+(d-1)n+\sum_{l=0}^{d-1}a_l = \mathcal{O}\left(nd+\sum_{l=0}^{d-1}a_l\right)$.
    \end{itemize}
\end{proof}

\end{document}